\documentclass[10pt,a4paper]{article}
\usepackage[utf8]{inputenc}
\usepackage[T1]{fontenc}
\usepackage{amsmath}
\usepackage{amsthm}
\usepackage{amsfonts}
\usepackage{amssymb}
\usepackage{graphicx}
\usepackage{xcolor}
\usepackage{enumerate}
\usepackage[labelformat=simple]{subcaption}

\usepackage{tikz}
\usetikzlibrary{calc,arrows,decorations.pathmorphing,decorations.markings,decorations.pathreplacing,snakes}
\tikzstyle{element}=[circle,draw=blue,fill=blue,inner sep=0pt,minimum size=3]

\title{Walking on Words\\
(Extended version)}
\author{Ian Pratt-Hartmann\\
	\ \\
	\ \\
Department of Computer Science\\
University of Manchester,
Manchester M13 9PL, United Kingdom\\
\ \\
Instytut Informatyki\\
Uniwersytet Opolski,
45-052 Opole, Poland}
\date{}

\newcommand{\ignore}[1]{}
\newcommand{\ti}[1]{\ensuremath{\tilde{#1}}}

\newcommand{\proofPara}{\medskip\noindent}

\newtheorem{theorem}{Theorem}

\newtheorem{lemma}[theorem]{Lemma}
\newtheorem{corollary}[theorem]{Corollary}
\newtheorem{claim}[theorem]{Claim}
\newcommand{\fA}{\ensuremath{\mathfrak{A}}} 
 
\newcommand{\Z}{\ensuremath{\mathbb{Z}}} 

\newcommand{\set}[1]{\ensuremath{\{ #1 \}}} 
\renewcommand{\phi}{\varphi}

\newcommand{\x}{\ensuremath{\texttt{\rm x}}}    

\def\presuper#1#2
  {\mathop{}%
   \mathopen{\vphantom{#2}}^{#1}%
   \kern-\scriptspace%
   #2}

\definecolor{GreenYellow}   {cmyk}{0.15,0,0.69,0}
	\definecolor{Yellow}        {cmyk}{0,0,1,0}
	\definecolor{Goldenrod}     {cmyk}{0,0.10,0.84,0}
	\definecolor{Dandelion}     {cmyk}{0,0.29,0.84,0}
	\definecolor{Apricot}       {cmyk}{0,0.32,0.52,0}
	\definecolor{Peach}         {cmyk}{0,0.50,0.70,0}
	\definecolor{Melon}         {cmyk}{0,0.46,0.50,0}
	\definecolor{YellowOrange}  {cmyk}{0,0.42,1,0}
	\definecolor{Orange}        {cmyk}{0,0.61,0.87,0}
	\definecolor{BurntOrange}   {cmyk}{0,0.51,1,0}
	\definecolor{Bittersweet}   {cmyk}{0,0.75,1,0.24}
	\definecolor{RedOrange}     {cmyk}{0,0.77,0.87,0}
	\definecolor{Mahogany}      {cmyk}{0,0.85,0.87,0.35}
	\definecolor{Maroon}        {cmyk}{0,0.87,0.68,0.32}
	\definecolor{BrickRed}      {cmyk}{0,0.89,0.94,0.28}
	\definecolor{Red}           {cmyk}{0,1,1,0}
	\definecolor{OrangeRed}     {cmyk}{0,1,0.50,0}
	\definecolor{RubineRed}     {cmyk}{0,1,0.13,0}
	\definecolor{WildStrawberry}{cmyk}{0,0.96,0.39,0}
	\definecolor{Salmon}        {cmyk}{0,0.53,0.38,0}
	\definecolor{CarnationPink} {cmyk}{0,0.63,0,0}
	\definecolor{Magenta}       {cmyk}{0,1,0,0}
	\definecolor{VioletRed}     {cmyk}{0,0.81,0,0}
	\definecolor{Rhodamine}     {cmyk}{0,0.82,0,0}
	\definecolor{Mulberry}      {cmyk}{0.34,0.90,0,0.02}
	\definecolor{RedViolet}     {cmyk}{0.07,0.90,0,0.34}
	\definecolor{Fuchsia}       {cmyk}{0.47,0.91,0,0.08}
	\definecolor{Lavender}      {cmyk}{0,0.48,0,0}
	\definecolor{Thistle}       {cmyk}{0.12,0.59,0,0}
	\definecolor{Orchid}        {cmyk}{0.32,0.64,0,0}
	\definecolor{DarkOrchid}    {cmyk}{0.40,0.80,0.20,0}
	\definecolor{Purple}        {cmyk}{0.45,0.86,0,0}
	\definecolor{Plum}          {cmyk}{0.50,1,0,0}
	\definecolor{Violet}        {cmyk}{0.79,0.88,0,0}
	\definecolor{RoyalPurple}   {cmyk}{0.75,0.90,0,0}
	\definecolor{BlueViolet}    {cmyk}{0.86,0.91,0,0.04}
	\definecolor{Periwinkle}    {cmyk}{0.57,0.55,0,0}
	\definecolor{CadetBlue}     {cmyk}{0.62,0.57,0.23,0}
	\definecolor{CornflowerBlue}{cmyk}{0.65,0.13,0,0}
	\definecolor{MidnightBlue}  {cmyk}{0.98,0.13,0,0.43}
	\definecolor{NavyBlue}      {cmyk}{0.94,0.54,0,0}
	\definecolor{RoyalBlue}     {cmyk}{1,0.50,0,0}
	\definecolor{Blue}          {cmyk}{1,1,0,0}
	\definecolor{Cerulean}      {cmyk}{0.94,0.11,0,0}
	\definecolor{Cyan}          {cmyk}{1,0,0,0}
	\definecolor{ProcessBlue}   {cmyk}{0.96,0,0,0}
	\definecolor{SkyBlue}       {cmyk}{0.62,0,0.12,0}
	\definecolor{Turquoise}     {cmyk}{0.85,0,0.20,0}
	\definecolor{TealBlue}      {cmyk}{0.86,0,0.34,0.02}
	\definecolor{Aquamarine}    {cmyk}{0.82,0,0.30,0}
	\definecolor{BlueGreen}     {cmyk}{0.85,0,0.33,0}
	\definecolor{Emerald}       {cmyk}{1,0,0.50,0}
	\definecolor{JungleGreen}   {cmyk}{0.99,0,0.52,0}
	\definecolor{SeaGreen}      {cmyk}{0.69,0,0.50,0}
	\definecolor{Green}         {cmyk}{1,0,1,0}
	\definecolor{ForestGreen}   {cmyk}{0.91,0,0.88,0.12}
	\definecolor{PineGreen}     {cmyk}{0.92,0,0.59,0.25}
	\definecolor{LimeGreen}     {cmyk}{0.50,0,1,0}
	\definecolor{YellowGreen}   {cmyk}{0.44,0,0.74,0}
	\definecolor{SpringGreen}   {cmyk}{0.26,0,0.76,0}
	\definecolor{OliveGreen}    {cmyk}{0.64,0,0.95,0.40}
	\definecolor{RawSienna}     {cmyk}{0,0.72,1,0.45}
	\definecolor{Sepia}         {cmyk}{0,0.83,1,0.70}
	\definecolor{Brown}         {cmyk}{0,0.81,1,0.60}
	\definecolor{Tan}           {cmyk}{0.14,0.42,0.56,0}
	\definecolor{Gray}          {cmyk}{0,0,0,0.50}
	\definecolor{Black}         {cmyk}{0,0,0,1}
	\definecolor{White}         {cmyk}{0,0,0,0}

\tikzset{->-/.style={decoration={
  markings,
  mark=at position #1 with {\arrow{>}}},postaction={decorate}}}
\tikzset{-<-/.style={decoration={
  markings,
  mark=at position #1 with {\arrow{<}}},postaction={decorate}}}
  
\tikzset{-|>-/.style={decoration={
    markings,
    mark=at position #1 with {\arrow{open triangle 60}}},postaction={decorate}}}

\begin{document}
\maketitle
\begin{abstract}
	Any function $f$ with domain $\set{1, \dots, m}$ and co-domain $\set{1, \dots, n}$ induces a natural map from words of length $n$ to those of length $m$: the $i$th letter of the output word ($1 \leq i \leq m$) is given by the $f(i)$th letter of the input word. We study this map in the case where $f$ is a surjection satisfying the condition $|f(i{+}1) {-} f(i)| \leq 1$ for $1 \leq i < m$. Intuitively,
	we think of $f$ as describing a `walk' on a word $u$, visiting every position, and yielding a word $w$ as the sequence of letters encountered {\em en route}. If such an $f$ exists, we say that $u$ {\em generates} $w$. Call a word {\em primitive} if it is not generated by any word shorter than itself.
	We show that every word  has, up to reversal, a unique {primitive generator}. Observing that, if a word contains a non-trivial palindrome, it can generate
	the same word via essentially different walks, we obtain conditions under which, for a chosen pair of walks $f$ and $g$, those walks yield the same word when applied to a given primitive word. Although the original impulse for studying primitive generators comes from their application to decision procedures in logic, we end, by way of further motivation, with an analysis of the primitive generators for certain word sequences defined via morphisms.  
\end{abstract}

\section{Introduction}
\label{sec:intro}	
Take any word over some alphabet, and, if it is non-empty, go to any letter in that word. Now repeat the following any number 
of times (possibly zero): 
scan the current letter, and print it out;
then either remain at the current letter, or move one letter to the left (if possible) or move one letter to the right (if possible). In effect, we are \textit{going for a walk on} the input word. 
Since any unvisited prefix or suffix of the input word cannot influence the word we print out, they may as well be ablated; letting $u$ be the factor of the input word comprising the scanned letters, and $w$ the word printed out, we say that $u$ \textit{generates} $w$. 
It is obvious that every word generates itself and its reversal, and that all other words it generates are strictly longer than itself. 
We ask about the converse of generation. Given a word $w$, what words $u$ generate it? 
Call a word {\em primitive} if it is not generated by any word shorter than itself. It is easy to see that every word must have a
generator that is itself primitive. We show that this {\em primitive generator} is in fact unique up to reversal. On the other hand, while
primitive generators are unique, the generating walks need not be, and this leads us to ask,
for a given pair of walks, whether we can characterize those primitive words $u$ for which they
output the same word $w$. We answer this question in terms of the palindromes contained in $u$. Specifically, 
for a primitive word $u$, the locations and lengths of the palindromes it contains determine which pairs of walks yield identical
outputs on $u$. As an illustration of the naturalness of the notion of primitive generator, we consider word sequences 
over the alphabet $\set{1, \dots, k}$ generated by the generalized Rauzy morphism $\sigma$, which maps the letter
$k$ to the word $1$, and any other letter $i$ ($1 \leq i <k$) to the two-letter word $1 \cdot (i{+}1)$. Setting 
$\alpha^{(k)}_1 = 1$ and $\alpha^{(k)}_{n+1} = \sigma(\alpha^{(k)}_n)$ for all $n \geq 1$, we obtain the
word sequence $\set{\alpha^{(k)}_n}_{n\geq 1}$.
We show that every word in this sequence from the $k$th onwards has the same primitive generator.


\section{Principal results}
\label{sec:results}	
Fix some alphabet $\Sigma$. We use $a, b, c \dots$ to stand for letters of $\Sigma$, and $u,v,w, \dots$ to stand for words over $\Sigma$.
The concatenation of two words $u$ and $v$ is denoted $uv$, or sometimes, for clarity, $u \cdot v$.  
For integers $i$, $k$ we write $[i,k]$ to mean the set $\set{j \in \Z \mid i \leq j \leq k}$.
If $u = a_1 \cdots a_n$ is a  
(possibly empty) word over $\Sigma$, and $f\colon [1,m] \rightarrow [1,n]$ is a function, we write 
$u^f$ to denote the word $a_{f(1)} \cdots a_{f(m)}$ of length $m$. We think of $f$ as telling us where in the word $u$ we should be at each 
time point in the interval $[1,m]$. 
Define a {\em walk} to be a surjection $f\colon [1,m] \rightarrow [1,n]$ 
satisfying $|f(i{+}1) {-} f(i)| \leq 1$  for all $i$ ($1 \leq i < m$).  
These conditions ensure that, as we move through the letters $a_{f(1)} \cdots a_{f(m)}$, we never 
change our position in $u$ by more than one letter at a time, and we visit every position of $u$ at least once. 
If $w = u^f$ for $f$ a walk, we say that $u$ {\em generates} $w$.
We may picture a walk as a piecewise linear function,
with the generat{\em ed} word superimposed on the abscissa and the generat{\em ing} word on the ordinate.
Fig.~\ref{fig:example} shows how
$u =$ {cbadefgh} generates $w=$ {abcbaaadefedadefghgf}. 
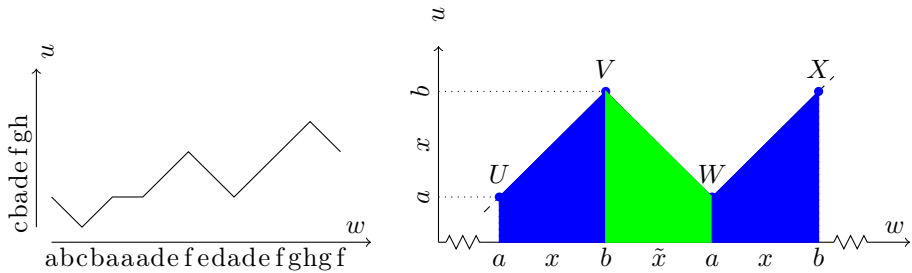
\begin{figure}
	\begin{subfigure}[t]{0.40\textwidth}
		\begin{center}
			
			\vspace{0.5cm}
			
			\begin{tikzpicture}[scale= 0.2]
				\draw[->] (0,0) --  (21,0);
				\coordinate[label={$w$}] (wLabel) at (20,0);
				
				\draw[->-=1] (-1,1) to (-1, 11.5);
				\coordinate[label={\rotatebox{90}{$u$}}] (wLabel) at(-0.25, 11.5) ;
				
				\coordinate[label=left:{\rotatebox{90}{c}}] (wLabel) at (-1,1);
				\coordinate[label=left:{\rotatebox{90}{b}}] (wLabel) at (-1,2);
				\coordinate[label=left:{\rotatebox{90}{a}}] (wLabel) at (-1,3);
				\coordinate[label=left:{\rotatebox{90}{d}}] (wLabel) at (-1,4);
				\coordinate[label=left:{\rotatebox{90}{e}}] (wLabel) at (-1,5);
				\coordinate[label=left:{\rotatebox{90}{f}}] (wLabel) at (-1,6);
				\coordinate[label=left:{\rotatebox{90}{g}}] (wLabel) at (-0.7,7);
				\coordinate[label=left:{\rotatebox{90}{h}}] (wLabel) at (-1,8);
				
				\coordinate[label={a}] (wLabel) at (0,-2.2);
				\coordinate[label={b}] (wLabel) at (1,-2.2);
				\coordinate[label={c}] (wLabel) at (2,-2.2);
				\coordinate[label={b}] (wLabel) at (3,-2.2);
				\coordinate[label={a}] (wLabel) at (4,-2.2);
				\coordinate[label={a}] (wLabel) at (5,-2.2);
				\coordinate[label={a}] (wLabel) at (6,-2.2);
				\coordinate[label={d}] (wLabel) at (7,-2.2);
				\coordinate[label={e}] (wLabel) at (8,-2.2);
				\coordinate[label={f}] (wLabel) at (9,-2.2);
				\coordinate[label={e}] (wLabel) at (10,-2.2);
				\coordinate[label={d}] (wLabel) at (11,-2.2);
				\coordinate[label={a}] (wLabel) at (12,-2.2);
				\coordinate[label={d}] (wLabel) at (13,-2.2);
				\coordinate[label={e}] (wLabel) at (14,-2.2);
				\coordinate[label={f}] (wLabel) at (15,-2.2);
				\coordinate[label={g}] (wLabel) at (16,-2.5);
				\coordinate[label={h}] (wLabel) at (17,-2.2);
				\coordinate[label={g}] (wLabel) at (18,-2.5);
				\coordinate[label={f}] (wLabel) at (19,-2.2);
				
				\draw (0,3) -- (2,1) -- (4,3) -- (6,3) -- (9,6) -- (12,3) --(17,8) -- (19,6);
			\end{tikzpicture}
		\end{center}
		\caption{Example of generation.}
		\label{fig:example}
	\end{subfigure}
\hspace{2mm}	
	\begin{subfigure}[t]{0.5\textwidth}
		\begin{center}
			\begin{tikzpicture}[scale= 0.2]
				\draw (3,0) to (26,0);
				\draw (0,0) -- (0.5,0);
				\draw[snake=zigzag,segment length = 5] (26,0) -- (28.5,0);
				\draw[->-=1] (28.5,0) -- (31,0);
				\draw[snake=zigzag,segment length = 5] (0.5,0) -- (3,0);
				\coordinate[label={$w$}] (wLabel) at (30,0);

				\draw[->-=1] (0,0) to (0, 13);
				\coordinate[label={\rotatebox{90}{$u$}}] (wLabel) at(0, 14) ;
				
				\coordinate (U) at (4,3);
				\coordinate (V) at (11,10);
				\coordinate (W) at (18,3);
				\coordinate (X) at (25,10);
				
				\coordinate [label=$a$] (Uw) at (4,-2.2);
				\coordinate [label=$x$] (UVw) at (7.5,-2.2);
				\coordinate [label=$b$] (Vw) at (11,-2.2);
				\coordinate [label=$\tilde{x}$] (VWw) at (14.5,-2.2);
				\coordinate [label=$a$] (Ww) at (18,-2.2);
				\coordinate [label=$x$] (WXw) at (21.5,-2.2);
				\coordinate [label=$b$] (Xw) at (25,-2.2);

				\draw[dashed] ($(U) - (1,1)$) -- (U);
				\draw (U) -- (V) -- (W) -- (X);
				\draw[dashed] (X) -- ($(X) + (1,1)$);

				\draw[dotted] (0,3) -- (4,3);
				\draw[dotted] (0,10) -- (11,10);
				\coordinate [label=left:{\rotatebox{90}{$a$}}] (aVert) at (0,3);
				\coordinate [label=left:{\rotatebox{90}{$x$}}] (aVert) at (0,6.5);
				\coordinate [label=left:{\rotatebox{90}{$b$}}] (aVert) at (0,10);

				\draw[<->,dotted]  ($ (U) + (0.2, 0) $) to node[label=above:{$\ell$}] {} ($ (U) + (6.8,0) $); 
				\draw[<->,dotted]  ($ (U) + (7.2, 0) $) to node[label=above:{$\ell$}] {} ($ (U) + (13.8,0) $); 
				\draw[<->,dotted]  ($ (W) + (0.2, 0) $) to node[label=above:{$\ell$}] {} ($ (W) + (6.8,0) $); 
				\draw[dotted]  (U) to ($ (U) - (0,3) $); 
				\draw[dotted]  (V) to ($ (V) - (0,10) $);
				\draw[dotted]  (W) to ($ (W) - (0,3) $);  
				\draw[dotted]  (X) to ($ (X) - (0,10) $);

				\node (nU) at (U)  [element,label=above:{$U$}] {}; 
				\node (nV) at (V)  [element,label=above:{$V$}] {}; 
				\node (nW) at (W)  [element,label=above:{$W$}] {}; 
				\node (nX) at (X)  [element,label=above:{$X$}] {}; 
				
				\filldraw[fill=blue,color=blue,opacity=0.3] (4,0) -- (4,3) -- (11,10) -- (11,0) --cycle;
				\filldraw[fill=green,color=green,opacity=0.3] (11,0) -- (11,10) -- (18,3) -- (18,0) --cycle;
				\filldraw[fill=blue,color=blue,opacity=0.3] (18,0) -- (18,3) -- (25,10) -- (25,0) --cycle;
			\end{tikzpicture}
		\end{center}
		\caption{The minimal leg $J= [V,W]$ of a walk.}
		\label{fig:minimalLeg}
	\end{subfigure}
	\label{fig:multifig1}
	\caption{Generation and minimal legs.}
\end{figure}

If $u= a_1 \cdots a_m$ is a word, denote the length of $u$ by $|u| =m$, and the reversal of $u$ by $\tilde{u}= a_m \cdots a_1$. 
Generation is evidently reflexive and reverse-reflexive: every word generates both itself and its reversal. Moreover, if $u$ generates $w$, then
$|u| \leq |w|$; in fact, $u$ and $\tilde{u}$ are the only words of length $|u|$ generated by $u$. 
We call $u$ \textit{primitive} if it is not generated by any word shorter than itself---equivalently, if it is generated only by itself and its reversal. 
For example, $babcd$ and $abcbcd$ are not primitive, because they are generated by $abcd$; but
$abcbda$ is primitive. 
Since the composition of two walks is a walk, generation is transitive: if $u$ generates $v$ and $v$ generates $w$, then $u$ generates $w$.
Define  a {\em primitive generator} of $w$ to be
a primitive word that generates $w$.
It follows easily from the above remarks that every word $w$ has some primitive generator $u$, and indeed, $\tilde{u}$ as well,
since the reversal of a primitive generator of $w$ is obviously also a primitive generator of $w$. The principal result of this paper is that there are no others:
\begin{theorem}
	The primitive generator of any word is unique up to reversal.	
	\label{theo:main}
\end{theorem}
As an immediate consequence, if $u$ is the primitive generator of $w$, and $v$ generates $w$, then $u$ generates $v$. 
Theorem~\ref{theo:main} is relatively surprising: let $u$ and $v$ be primitive words. Now suppose we go for a walk on $u$ and, independently, go for a walk on $v$; recalling the stipulation that the two walks visit every position in the words they apply to, the theorem says that,
provided only that
$u \neq v$ and $u \neq \tilde{v}$, there is no possibility of coordinating
these walks so that the sequences of visited letters are the same. 

A \textit{palindrome} is a word $u$ such that $u = \tilde{u}$; a \textit{non-trivial} palindrome is one of length at least 2.
If $u$ is a non-trivial palindrome, then it is not primitive. Indeed, if $|u|$ is even, then $u$ has a double
letter in the middle, and so is certainly not primitive (it is generated by the word in which one of the occurrences of the doubled letter is deleted); if $|u|$ is odd, then it is generated by its prefix of length $(|u|{+}1)/2 < |u|$ 
(start at the beginning, go just over half way, then return to the start). 
Call a word {\em uniliteral} if it is of the form $a^n$ for some letter $a$ and some $n \geq 0$. Note that the empty word $\epsilon$ counts as uniliteral.
\begin{corollary}
	Every uniliteral word has precisely one primitive generator; all others have precisely two.	
	\label{cor:main}
\end{corollary}
\begin{proof}
	By Theorem~\ref{theo:main}, if $w$ is any word, its primitive generators are of the form $u$ and $\tilde{u}$ for some word $u$.
	The first statement of the corollary is obvious: if $w = \epsilon$ then $u = \tilde{u}=\epsilon$; and if $u= a^n$ for some $n$ ($n \geq 1$), then  $u = \tilde{u}=a$.
	If $w$ is not uniliteral, then $|u| > 1$. But since non-trivial palindromes cannot be primitive, $u \neq \tilde{u}$.
\end{proof}

Yet another way of stating Theorem~\ref{theo:main} is to say that, for any fixed word $w$, the equation $u^f = w$ has exactly one primitive solution for $u$, up to reversal. The same is not true, however,  of solutions for $f$, even if we fix the choice of primitive generator (either $u$ or $\tilde{u}$). Indeed, 
$u= abcbd$ is one of the two
primitive generators of $w = abcbcbd$, but we have $u^f= w$ for $f \colon[1,7] \rightarrow [1,5]$ given by either of the courses of values $[1,2,3,4,3,4,5]$ or $[1,2,3,2,3,4,5]$.
Let $u$ be a primitive word. Say that $u$ is {\em perfect} if $u^f = u^g$ implies $f = g$ for any walks $f$ and $g$ on $u$. Thus, $abcbd$ is primitive but not
perfect. On the other hand, it is easy to characterize those primitive words that are perfect:
\begin{theorem}
	Let $u$ be a word. Then $u$ is perfect if and only if it contains no non-trivial palindrome as a factor. 
	\label{theo:noble}
\end{theorem}

Theorem~\ref{theo:noble} tells us that generating walks are uniquely determined as long as the primitive generator $u$ does not contain a non-trivial palindrome, but gives us little information if $u$ does contain a non-trivial palindrome. In that case, we would like a characterization of \textit{which} pairs of walks on $u$ yield identical words. We answer this question in terms of the \textit{positions} of the palindromes contained in $u$. Let $u = a_1 \cdots a_n$ be a word. We denote the $i$th letter of $u$ by $u[i] = a_i$, and the factor of $u$ from the
$i$th to $j$th letters by $u[i,j] = a_i \cdots a_j$. If $u[i,j]$ is a non-trivial palindrome, call the ordered pair $\langle i,j\rangle$ a \textit{defect}
of $u$, and denote the set of defects of $u$ by $\Delta(u)$. Regarding 
$\Delta(u)$ as a binary relation on the set $[1,n]$, we write 
$\Delta^*(u)$ for its equivalence closure, the smallest reflexive, symmetric and transitive relation including $\Delta(u)$. The interplay between 
defects and walks is then summed up in the following theorem.

\begin{theorem} 
	Let ${u}$ be a primitive word of length $n$, and $f$, $g$ walks with domain $[1, m]$ and co-domain $[1, n]$.
	Then 
	$u^f = u^g$ if and only if
	$\langle f(i), g(i) \rangle \in \Delta^*(u)$ for all $i \in [1,m]$. 
	\label{theo:defects}
\end{theorem}

The motivation for the study of primitive generators comes from the study of the decision problem in (fragments of) first-order logic, in presentations
where the logical variables are taken to be $x_1, x_2, \dots$, and all signatures are assumed to be purely relational.
Call a first-order formula $\phi$ \textit{index-normal} if, for any quantified sub-formula $Qx_k \psi$ of $\phi$, $\psi$ is a
Boolean combination of formulas that are either atomic with free variables among $x_1$, \dots , $x_k$, or have as their major connective a quantifier 
binding $x_{k+1}$. By re-indexing variables, any first-order formula can easily be written as a logically equivalent index-normal formula.
We call an index-normal formula {\em adjacent} if, in any atomic sub-formula, the indices of neighbouring arguments never differ by more than 1. For example,
an atomic sub-formula $p(x_4,x_4,x_5,x_4,x_3)$ is allowed, but an atomic sub-formula $q(x_1,x_3)$ is not. 
It was shown in~\cite{words:bkp-h23} that the problem of determining validity for adjacent formulas is decidable.
A key notion in analysing this fragment is that of  an 
{\em adjacent type}. Let $\fA$ be a structure interpreting some relational signature, and $\bar{a}$ a tuple of elements from its domain, $A$. Define the {\em adjacent type of} $\bar{a}$ (\textit{in} $\fA$) to be the set of all adjacent atomic formulas $q(\bar{x})$ satisfied by $\bar{a}$ in $\fA$. If we think now of $\bar{a}$ as a word over the (possibly infinite) alphabet $A$, 
it can easily be shown that the adjacent type of $\bar{a}$ is determined by the adjacent type of its primitive generator. Thus, models of formulas can be unambiguously constructed by specifying only the adjacent types of \textit{primitive} tuples, a crucial technique in establishing decidability of the satisfiability problem.

Notwithstanding its logical genealogy, the concept of primitive generator may be of interest in its own right within the field of string combinatorics. 
For illustration, consider the sequences of words $\set{\alpha^{(k)}_n}_{n\geq 1}$ over the alphabet $\set{1, \dots, k}$, defined by setting $\alpha^{(k)}_1= 1$ 
and $\alpha^{{(k)}}_{n+1} = \sigma(\alpha^{(k)}_n)$, where $\sigma \colon \set{1, \dots, k}^* \rightarrow \set{1, \dots, k}^*$ is the monoid endomorphism 
given by
\begin{equation*}  
	\sigma(i) =  
	\begin{cases}
		1 \cdot (i+1) &  \text{if $i <k$}\\
		1 & \text{if $i = k$.}
	\end{cases}
\end{equation*}
(Here, the operator $\cdot$ represents string concatenation, not integer multiplication!) For $k = 2$, we obtain the so-called \textit{Fibonacci word sequence} 1, 12, 121, 12112, \dots; 
for $k = 3$, we obtain the \textit{tribonacci word sequence} 1, 12, 1213, 1213121, \dots; and so on. 
A simple induction shows that, for all $k \geq 2$ and 
all $n >k$, $\alpha^{(k)}_n = \alpha^{(k)}_{n-1}\alpha^{(k)}_{n-2} \cdots \alpha^{(k)}_{n-k}$.
In other words, each element of the sequence $\set{\alpha^{(k)}_n}_{n\geq 1}$ after the $k$th is the concatenation, in reverse order, of the previous $k$ elements; for this reason, the word sequence obtained is referred to as the $k$-\textit{bonacci word sequence}.	
A simple proof also shows that $\alpha^{(k)}_{n}$ is always a left-prefix of $\alpha^{(k)}_{n+1}$, so that we may speak of the infinite word $\omega^{(k)}$ defined by taking the limit $\lim_{n \rightarrow \infty} \alpha^{(k)}_n$ in the obvious sense. 
Thus, the infinite word $\omega^{(2)}= 12112\cdots$ is the
(infinite) {\em Fibonacci word}, and 
$\omega^{(3)}= 1213121\cdots$ the (infinite) {\em tribonacci word}.
The Fibonacci word is an example
of a \textit{Sturmian word} (see, e.g.~\cite[Ch.~6]{words:fogg03} for an extensive treatment). The morphism yielding the tribonacci word is 
sometimes called the \textit{Rauzy morphism}~\cite[p.~149]{words:rauzy82} (see also~\cite[Secs.~10.7 and 10.8]{words:lothaire05}).  Intriguingly, for a fixed $k$, all but the first $k$ elements of $\set{\alpha^{(k)}_n}_{n\geq 1}$ share the same primitive generator:
\begin{theorem}
	For all $k \geq 2$, there exists a word $\gamma_k$ such that, for all $n \geq k$, $\gamma_k$ is the primitive generator of $\alpha^{(k)}_n$. 
	\label{theo:kbonacci}
\end{theorem}
The proof of Theorem~\ref{theo:kbonacci} exploits a feature of the words $\alpha^{(k)}_n$ that is obvious when one computes a few examples: they are riddled with palindromes. As one might then expect in view of Theorem~\ref{theo:defects}, for all $k$ and all $n \geq k$, the primitive generator $\gamma_k$  generates  $\alpha^{(k)}_n$ via many different walks---in fact via walks beginning at any position of $\gamma_k$ occupied by the letter 1.

\section{Uniqueness of primitive generators}
\label{sec:main}	
The following terminology will be useful. (Refer to Fig.~\ref{fig:example} for motivation.)
Let $f \colon [1,m] \rightarrow [1,n]$ be a walk, with $m >1$.
By a {\em leg} of $f$, we mean a maximal interval $[i,j] \subseteq [1,m]$ such that, for $h$ in the range $i \leq h < j$, the difference $d= f(h{{+}}1) {-} f(h)$ is constant. We speak of a {\em descending}, {\em flat} or {\em ascending} leg, depending on whether $d$ is $-1$, 0 or 1. The {\em length} of the leg is $j{-}i$.
A leg $[i,j]$ is {\em initial} if $i= 1$, {\em final} if $j= m$, {\em terminal} if it is either initial or final, and {\em internal} if it is not terminal. 
A number $h$ which forms the boundary between two consecutive legs will be called a {\em waypoint}. We count the numbers $1$ and $m$ as waypoints by courtesy, and refer to them as {\em terminal waypoints}; all other waypoints are \textit{internal}. 
Thus, a walk consists of a sequence of legs from one waypoint to another.
If $h$ is an internal waypoint where the change is from an increasing to a decreasing leg, we call $h$ a {\em peak}; if the change is from a decreasing to an increasing leg, we call it a {\em trough}. Not all waypoints need be peaks or troughs, because some legs may be flat; however, it is these waypoints that will chiefly concern us in the sequel. 

\begin{lemma}
	A word $u$ is not primitive if and only if it is of any of the following forms, where $a$, $b$ are letters and $x$, $y$, $z$ are words:
	\textup{(i)} $xaay$,
	\textup{(ii)} $b\tilde{x}axby$,
	\textup{(iii)} $yb\tilde{x}axb$ or
	\textup{(iv)} $yaxb\tilde{x}axbz$.
	\label{lma:main}
\end{lemma}
\begin{proof}
	We prove the if-direction by considering the cases in turn:
(i) $xay$ generates $xaay$ (traverse $xay$ from beginning to end, pausing one time step on the $a$); (ii) $axby$ generates $b\tilde{x}axby$ (start at the $b$, go left to the beginning and then go to the end);
(iii) $yb\tilde{x}a$ generates $yb\tilde{x}axb$ (go from beginning to end, turn back and stop at the $b$); and (iv) $yaxbz$ generates $yaxb\tilde{x}axbz$
(go to the $b$, then turn back to the $a$, then turn again and go to the end). For the converse, suppose that $w$ is not primitive,
let $u$ be a generator of $w$ with $|u| < |w|$, and let 
$f$ be a walk such that $w= u^f$. Since $|u| < |w|$, $f$ has at least two legs. 
If $f$ has any flat legs, then $w$ contains a repeated letter and so is of the form $xaay$, yielding Case (i). Otherwise, all legs of $f$ are either increasing or decreasing. Take a shortest leg, say $J \subseteq [1,m]$.
If $J$ is internal (Fig.~\ref{fig:minimalLeg}), then $w$ is of the form $yaxb\tilde{x}axbz$, yielding Case (iv).
Similarly, if $J$ is initial, then $w$ is of the form $b\tilde{x}axby$, yielding Case (ii), and
if $J$ is final, then $w$ is of the form $yb\tilde{x}axb$, yielding Case (iii). 
\end{proof}

In the sequel, we shall primarily employ the if-direction of Lemma~\ref{lma:main}.
It easily follows from Cases (i) and (ii) of Lemma~\ref{lma:main} that, over the alphabet $\set{1,2}$, there are exactly five primitive words:
$\epsilon$, 1, 2, 12, and 21. However, over any larger alphabet, there are infinitely many. For example, over the alphabet $\set{1,2,3}$, 
the set of primitive words is easily seen to be given by the regular expression $[(\epsilon + 3 + 23)(123)^*(\epsilon + 1 + 12)]+[(\epsilon + 2 + 32)(132)^*(\epsilon + 1 + 13)]$. Over alphabets of any finite size, the set of primitive words is context-sensitive. This follows from the fact that the four patterns of Lemma~\ref{lma:main} define context-sensitive languages, together with the
standard Boolean closure properties of context-sensitive languages.

We shall occasionally need to consider a broader class of functions than walks.
Define a {\em stroll} to be a function $f\colon [1,m] \rightarrow [1,n]$  satisfying $|f(i{+}1) {-} f(i)| \leq 1$  for all $i$ ($1 \leq i < m$).  
Thus, a walk is a stroll which is surjective.
Let  $f \colon [1,m] \rightarrow [1,n]$ be a stroll. If $f(i) = f(j)$ for some $i$, $j$ ($1 < i < j < m$)  
define the function $f' \colon [1,m{-}j{+}i] \rightarrow [1,n]$ by setting
$f'(h) = f(h)$ if $1 \leq h \leq i$, and $f'(h) = f(h{+}j{-}i)$ otherwise.
Intuitively, $f'$ is just like $f$, but with the interval $[i,j{-}1]$---equivalently, the interval  $[i{+}1,j]$---removed. 
Evidently, $f'$ is a also a stroll, and 
we denote it by $f /  [i,j]$. For the cases $i = 1$ or $j = m$, we change the definition slightly, as no analogue of the condition $f(i) = f(j)$ is required. 
Specifically if $1 \leq i < j \leq m$, 
define the functions $f' \colon [1,m{-}j+1] \rightarrow [1,n]$ and $f'' \colon [1,i] \rightarrow [1,n]$ by $f'(h) = f(j{+}h{-}1)$ and
$f''(h) = f(h)$.
Intuitively, $f'$ is just like $f$, but with the interval $[1,j{-}1]$ removed, and
$f''$ is just like $f$, but with the interval $[i{+}1,m]$ removed.
Again $f'$ and $f''$ are also strolls,
and we denote them by $f /  [1,j]$ and $f /  [i,m]$, respectively.

\newtheorem*{Restatetheo:main}{Theorem~\ref{theo:main}}
\begin{Restatetheo:main}
	The primitive generator of any word is unique up to reversal.	
\end{Restatetheo:main}
\begin{proof}
	We proceed by contradiction, supposing that $u$ and $v$ are primitive words such that neither $u = v$ nor $u = \tilde{v}$, and $w$ is a word generated from
	$u$ by some walk $f$ and from $v$ by some walk $g$. Write $|w| = m$.
	Crucially, we may assume without loss of generality that $w$ is a {\em shortest counterexample}---that is, a shortest word 
	for which such $u$, $v$, $f$ and $g$ exist.  Observe that, since $u$ and $v$ are primitive, they feature
	no immediately repeated letter. So suppose $w$ does---i.e.~is of the form $w = xaay$ for some words $x$, $y$ and letter $a$. Letting $i = |x|{+}1$, we must therefore have
	$f(i)= f(i{+}1)$ and $g(i) = g(i{+}1)$. Now let $f'= f /  [i,i{+}1]$, $g'= g /  [i,i{+}1]$ and $w' = w[1,i]\cdot w[i{+}2,m]$.
	We see that $f'$ is surjective if $f$ is, and similarly for $g'$, and moreover  
	that $w' = u^{f'} = v^{g'}$, contrary to the assumption that $w$ is shortest. Hence $w$ contains no immediately repeated letters,
	whence all legs of $f$ and $g$ are either increasing or decreasing, and all internal waypoints are either peaks or troughs. 
	
	We claim first that at least one of $f$ or $g$ must have an internal waypoint.  For if not, we have $w = u$ or $w = \tilde{u}$ and $w = v$ or $w = \tilde{v}$,
	whence $u = v$ or $u = \tilde{v}$, contrary to assumption. It then follows that \textit{both} $f$ \textit{and} $g$ have an internal waypoint. For suppose $f$ has an internal waypoint (either a
	peak or a trough); then $w$ is not primitive. But if $g$ does not have an internal waypoint, $w = v$ or $w = \tilde{v}$, contrary to the assumption that $v$ is primitive.

	We use upper case letters in the sequel  to denote integers in the range $[1,n]$ which are somehow significant for the walks $f$ or $g$: note that 
	these need not be
	waypoints.
	Let $\ell$ denote the minimal length of a leg on either of the walks $f$ or $g$. Without loss of generality, we may take this
	minimum to be achieved on a leg of $f$, say $[V,W]$.
	
	\proofPara
	We suppose for the present that this leg is \textit{internal}.  
	Fig.~\ref{fig:minimalLeg} illustrates this situation where
	$V$ is a peak and $W$ a trough; but nothing essential would change if it were the other way around.
	Write
	$U = V-\ell$ and $X = W {+} \ell$. By the minimality of $[V,W]$ (assumed internal), $U \geq 1$ and $X \leq m$; moreover, $f$ is monotone on $[U,V]$, $[V,W]$ and $[W,X]$. 
	Now let $w[U]= a$,  $w[V] = b$ and $w[U{+}1,V-1] = x$.  
	Since $V$ is a waypoint on $f$,  $w[W] = a$ and  $w[V{+}1,W-1] = \tilde{x}$. 
	Similarly,  $w[X] = b$ and 
	$w[W{+}1,X{-}1] = \tilde{\tilde{x}}=x$.
	We see immediately that $g$ must have a waypoint 
	in the interval $[U{+}1, X{-}1]$, for otherwise, $v$ (or $\tilde{v}$) contains a factor $axb\tilde{x}axb$,     
	contrary to the assumption that $v$ is primitive (Lemma~\ref{lma:main}, case (iv)). Let $Y$ be the 
	waypoint on $g$ which is closest to either of $V$ or $W$.  Replacing $w$ by its reversal if necessary, assume that $|Y {-} V| \leq |Y -W|$, and write  $k = |Y {-} V|$.
	We consider possible values of $k \in [0,\ell{-}1]$ 
	in turn, deriving a contradiction in each case. 
	
	\proofPara
	Case (i): $k = 0$  (i.e.~$Y = V$). 
	For definiteness, let us suppose that $Y$ is as a peak, rather than a trough, but the reasoning is entirely unaffected by this 
	determination. 
	By the minimality of the leg $[V,W]$, $g$ has no other waypoints in the interval $[U{+}1,W-1]$, and $g(U) = g(W)$. By inspection of Fig.~\ref{fig:minimalLeg}, it is also clear from the minimality of the leg
	$[V,W]$ 
	that $f'= f /  [U,W]$ is surjective (and hence a walk). We see immediately that the stroll
	$g'= g /  [U,W]$ is not surjective. Indeed, if it were, writing $w' = w[1,U] \cdot w[W{+}1,n]$, we would have
	${w'}= u^{f'} = v^{g'}$, contrary to the assumption that $w$ is a shortest counterexample.
	In other words, there are positions of $v$ which $g$ reaches over the range $[U{+}1, W-1])$ that it does not reach outside this range. It follows
	that the position $g(V) = g(Y)$ in the string $v$ is actually terminal. (Since we are assuming that $Y$ is a peak, $g(Y) = |v|$; but the following 
	reasoning is unaffected if $Y$ is a trough and $g(Y) = 1$.) It also follows that $W$ itself cannot be a waypoint of $g$. For otherwise, the leg following
	$W$, which is of length at least $\ell$, covers all values in $g([U, W])$, thus ensuring that $g'$ is surjective, which we have just shown to be false.
	However, $g$ must have some waypoint in $[V{+}1, X-1]$. 
	For if not, then $g$ is decreasing between $V$ and $X$ (remember that $g(Y)= g(V) =|v|$), and 
	thus $v$ has a suffix $b\tilde{x}axb$, contrary to the assumption that $v$ is primitive (Lemma~\ref{lma:main} case (iii)). 
	By the minimality of 
	the leg $[V,W]$ we see that there is exactly one such waypoint, say $Z$. Since we have already shown that $Y$ is the only waypoint on $g$ in $[U{+}1,W{-}1]$, and that $W$ is not a waypoint on $g$, it follows that  
	$Z \in [W{+}1, X{-}1]$.
	
	Now let $j = Z-W$. (Thus, $1 \leq j < \ell$.) If $j > \frac{1}{2}\ell$, we obtain the situation depicted in Fig.~\ref{fig:YVbigj}.
	\begin{figure}
		\begin{subfigure}[t]{0.5\textwidth}
			\begin{center}				
\resizebox{6cm}{!}{
\begin{tikzpicture}[scale= 0.2]
					\draw (3,0) to (26,0);
					\draw (0,0) -- (0.5,0);
					\draw[snake=zigzag,segment length = 5] (26,0) -- (28.5,0);
					\draw[->-=1] (28.5,0) -- (31,0);
					\draw[snake=zigzag,segment length = 5] (0.5,0) -- (3,0);
					\coordinate[label={$w$}] (wLabel) at (32,0);

					\draw[->-=1] (0,0) to (0, 12.5);
					\coordinate[label={\rotatebox{90}{$u$}}] (wLabel) at(-0.5, 10.5) ;
					
					\coordinate (U) at (4,3);
					\coordinate (V) at (11,10);
					\coordinate (W) at (18,3);
					\coordinate (X) at (25,10);
					
					\coordinate [label=$a$] (Uw) at (4,-2);
					\coordinate [label=$x$] (UVw) at (7.5,-2);
					\coordinate [label=$b$] (Vw) at (11,-2);
					\coordinate [label=$\tilde{x}$] (VWw) at (14.5,-2);
					\coordinate [label=$a$] (Ww) at (18,-2);
					\coordinate [label=$b$] (Ww) at (19,-2);
					\coordinate [label=$c$] (Xw) at (22,-2);	
					\coordinate [label=$b$] (Xw) at (25,-2);	
					
					\draw[dashed] ($(U) - (1,1)$) -- (U);
					\draw (U) -- (V) -- (W) -- (X);
					\draw[dashed] (X) -- ($(X) + (1,1)$);
					
					\draw[<->,dotted]  ($ (U) + (0.2, 0) $) to node[label=below:{$\ell$}] {} ($ (U) + (6.8,0) $); 
					\draw[<->,dotted]  ($ (U) + (7.2, 0) $) to node[label=below:{$\ell$}] {} ($ (U) + (13.8,0) $); 
					\draw[<->,dotted]  ($ (W) + (0.2, 0) $) to node[label=below:{$\ell$}] {} ($ (W) + (6.8,0) $); 
					\draw[dotted]  (U) to ($ (U) - (0,3) $); 
					\draw[dotted]  (V) to ($ (V) - (0,10) $);
					\draw[dotted]  ($(W) + (0,14) $) to ($ (W) - (0,3) $);  
					\draw[dotted]  (X) to ($ (X) - (0,10) $);

					\node (nU) at (U)  [element,label=above:{$U$}] {}; 
					\node (nV) at (V)  [element,label=above:{$V$}] {}; 
					\node (nW) at (W)  [element,label=above:{$W$}] {}; 
					\node (nX) at (X)  [element,label=above:{$X$}] {}; 
					
					
					\draw[->-=1] (0,15) to (0, 25);
					\coordinate[label={\rotatebox{90}{$v$}}] (wLabel) at(-0.5, 23) ;
					
					\coordinate (T) at (4, 17);
					\coordinate (Y) at (11,24);
					\coordinate (Z) at (22,13);
					\coordinate (ZZ) at (19,16);
					\coordinate (XX) at (25,16);
					
					\draw[<->,dotted]  (18.2,11) to node[label=below:{$j$}] {} (21.8,11);
					
					\node (nY) at (Y)  [element,label=above:{$Y$}] {}; 
					\node (nZ) at (Z)  [element,label=above:{$Z$}] {}; 
					
					\draw (T) -- (Y) -- (Z) -- (XX);	
					\draw[dashed] ($(T) - (1,1)$) -- (T);
					\draw[dashed] (XX) -- ($(XX) + (1,1)$);	
					
					\draw[dotted]  (ZZ) to (19,0); 
					\draw[dotted]  (Z) to (22,0);
					\draw[dotted]  (XX) to (25,0);  
					\draw[dotted]  (ZZ) to (XX);
					
					\coordinate [label=$ybzc\tilde{z}$] (WXw) at (21.5,-4.25);	
					\draw[<->,dotted] (18.4,-2) -- (24.6,-2);
					
					\filldraw[fill=blue,color=blue,opacity=0.3] (4,0) -- (4,17) -- (11,24) -- (11,0) --cycle;
					\filldraw[fill=green,color=green,opacity=0.3] (11,0) -- (11,24) -- (18,17) -- (18,0) --cycle;
					\filldraw[fill=blue,color=blue,opacity=0.3] (18,0) -- (18,3) -- (25,10) -- (25,0) --cycle; 
				\end{tikzpicture}
}
			\end{center}
			\caption{The condition $j= Z-W > \frac{1}{2}\ell$.}
			\label{fig:YVbigj}
		\end{subfigure}
	    \hspace{2mm}
		\begin{subfigure}[t]{0.5\textwidth}
			\begin{center}
\resizebox{6cm}{!}{
\begin{tikzpicture}[scale= 0.2]
					\draw (3,0) to (26,0);
					\draw (0,0) -- (0.5,0);
					\draw[snake=zigzag,segment length = 5] (26,0) -- (28.5,0);
					\draw[->-=1] (28.5,0) -- (31,0);
					\draw[snake=zigzag,segment length = 5] (0.5,0) -- (3,0);
					\coordinate[label={$w$}] (wLabel) at (32,0);

					\draw[->-=1] (0,0) to (0, 12.5);
					\coordinate[label={\rotatebox{90}{$u$}}] (wLabel) at(-0.5, 10.5) ;
					
					\coordinate (U) at (4,3);
					\coordinate (V) at (11,10);
					\coordinate (W) at (18,3);
					\coordinate (X) at (25,10);
					
					\coordinate [label=$a$] (Uw) at (4,-2);
					\coordinate [label=$x$] (UVw) at (7.5,-2);
					\coordinate [label=$b$] (Vw) at (11,-2);
					\coordinate [label=$a$] (Ww) at (18,-2);
					
					\draw[dashed] ($(U) - (1,1)$) -- (U);
					\draw (U) -- (V) -- (W) -- (X);
					\draw[dashed] (X) -- ($(X) + (1,1)$);
					
					\draw[dotted]  (U) to ($ (U) - (0,3) $); 
					\draw[dotted]  (V) to ($ (V) - (0,10) $);
					\draw[dotted]  ($(W) + (0,14) $) to ($ (W) - (0,3) $);  
					\draw[dotted]  (X) to ($ (X) - (0,10) $); 
					
					\draw[dotted] (15,6) -- (21,6);

					\node (nU) at (U)  [element,label=above:{$U$}] {}; 
					\node (nV) at (V)  [element,label=above:{$V$}] {}; 
					\node (nW) at (W)  [element,label=above:{$W$}] {}; 
					\node (nX) at (X)  [element,label=above:{$X$}] {}; 
					
					
					\draw[->-=1] (0,15) to (0, 25);
					\coordinate[label={\rotatebox{90}{$v$}}] (wLabel) at(-0.5, 23) ;
					
					\coordinate (T) at (4, 17);
					\coordinate (Y) at (11,24);
					\coordinate (Z) at (21,14);
					\coordinate (WW) at (12,23);
					\coordinate (ZZ) at (15,20);
					\coordinate (XX) at (24,17);
					
					\draw[<->,dotted]  (12.2,11) to node[label=below:{$j$}] {} (14.8,11);			
					\draw[<->,dotted]  (15.2,11) to node[label=below:{$j$}] {} (17.8,11);
					\draw[<->,dotted]  (18.2,11) to node[label=below:{$j$}] {} (20.8,11);
					\draw[<->,dotted]  (21.2,11) to node[label=below:{$j$}] {} (23.8,11);

					\node (nY) at (Y)  [element,label=above:{$Y$}] {}; 
					\node (nZ) at (Z)  [element,label=above:{$Z$}] {}; 
					
					\draw (T) -- (Y) -- (Z) -- (XX);	
					\draw[dashed] ($(T) - (1,1)$) -- (T);
					\draw[dashed] (XX) -- ($(XX) + (1,1)$);	
					
					\draw[dotted]  (WW) to (12,0); 
					\draw[dotted]  (ZZ) to (15,0); 
					\draw[dotted]  (Z) to (21,0);
					\draw[dotted]  (XX) to (24,0);  
					
					\draw[dotted]  ($ (XX) - (6,0) $) to (XX);
					
					\coordinate [label=$a$] (Vw) at (12,-2);	
					\coordinate [label=$y$] (Vw) at (13.5,-2.25);				
					\coordinate [label=$c$] (cw) at (15,-2);
					\coordinate [label=$\tilde{y}$] (Vw) at (16.5,-2.25);				
					\coordinate [label=$a$] (Vw) at (18,-2);	
					\coordinate [label=$y$] (Vw) at (19.5,-2.25);				
					\coordinate [label=$c$] (Zw) at (21,-2);
					\coordinate [label=$\tilde{y}$] (Vw) at (22.5,-2.25);				
					\coordinate [label=$a$] (Vw) at (24,-2);	
					
					\filldraw[fill=blue,color=blue,opacity=0.3] (15,0) -- (15,20) -- (12, 23) -- (12,0) -- cycle; 
					\filldraw[fill=green,color=green,opacity=0.3] (18,17) -- (18,0) -- (15,0) -- (15,20) -- cycle; 
					\filldraw[fill=blue,color=blue,opacity=0.3] (21,14) -- (18,17) -- (18,0) -- (21,0) --cycle;
					\filldraw[fill=green,color=green,opacity=0.3] (21,14) -- (21,0) -- (24,0) -- (24, 17) -- cycle;

				\end{tikzpicture}
}			
\end{center}
			\caption{The condition $j= Z-W < \frac{1}{2}\ell$.}
			\label{fig:YVsmallj}
		\end{subfigure}
		\caption{The walk $g$ has waypoints at $Y=V$ and at $Z$.}
		\label{fig:YVbigjAndYVsmallj}
		
	\end{figure}
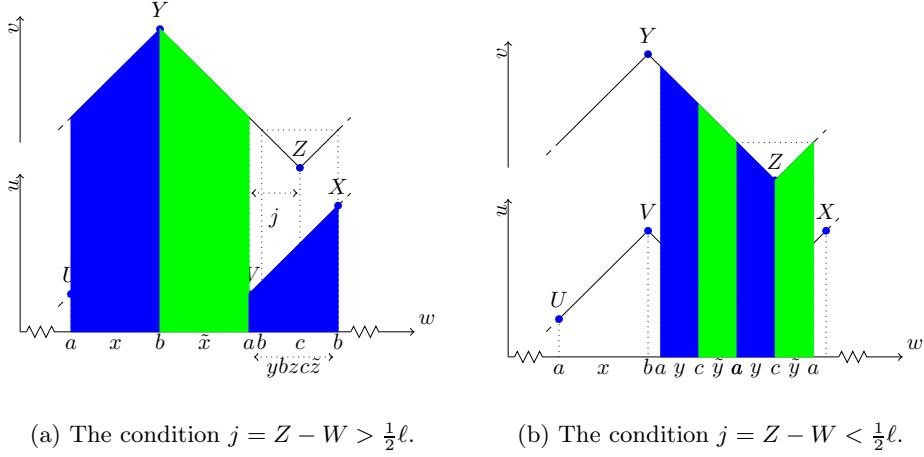
	Since $g$ has a waypoint at $Z$ and remembering that $w[W{+}1,X-1] = x$ and $w[X] =b$, we see that
	$x$ has the form $ybzc\tilde{z}$ for some strings $y$ and $z$ and some letter $c = w[Z]$. 
	But we also know that $g(V) = g(Y)= |v|$, the final position of $v$, 
	so that $v$ has a suffix $xb =   (ybzc\tilde{z})b$, and hence the suffix $bzc\tilde{z}b$, contrary to the assumption that $v$ is primitive (Lemma~\ref{lma:main}  case (iii)).
	Furthermore, if  $j = \frac{1}{2}\ell$, then, by the same reasoning, $x$ has the form $zc\tilde{z}$ and $a = b$.
	Again then, $v$ has a suffix $bzc\tilde{z}b$, contrary to the assumption that $v$ is primitive. 
	We conclude that $j < \frac{1}{2}\ell$, and we obtain the situation depicted in Fig.~\ref{fig:YVsmallj}.
	Now let $c= w[Z]$ and $y= w[W{+}1,Z{-}1]$. By considering the waypoint $Z$ on $g$, we see that $w[Z,Z{+}j] = c\tilde{y}a$, whence $w[W,W{+}2j]= ayc\tilde{y}a$.
	By considering the waypoint $W$ on $f$, we see that also $w[W{-}2j,W] = ayc\tilde{y}a$, whence $w[W{-}2j,W{+}j]= ayc\tilde{y}ayc$. But there are no waypoints of
	$g$ strictly between $V = Y < W{-}2j$ and $Z$, whence 
	$\tilde{v}$ contains the factor 
	$ayc\tilde{y}ayc$, 
	contrary to the supposition that $v$ is primitive (Lemma~\ref{lma:main} case (iv)). 
	
	\proofPara
	Case (ii): $1 \leq k \leq \frac{1}{3}\ell$. We may have either $Y > V$ or $Y < V$: Fig.~\ref{fig:close} shows the former case; however, the
	reasoning in the latter is almost identical. Let $w[V] = b$ and $w[Y]= c$. Furthermore, let $w[V{+}1, Y-1] = y$. Since $V$ is a waypoint of $f$, we have
	$w[V {-} k] = c$ and $w[V {-} k {+}1, V{-}1] = \tilde{y}$, whence $w[Y {-}2k, Y] = w[V {-}k, Y] = c\tilde{y}byc$.
	Since $Y$ is a waypoint of $g$, we have
	$w[Y, Y {+} 2k] = c\tilde{y}byc$, whence $w[V, V {+} 3k] = byc\tilde{y}byc$. And since $\ell \geq 3k$, there 
	is no waypoint on $f$ in the interval  $w[V {+} 1, V {+} 3k-1]$, whence
	$\tilde{u}$ contains the factor $byc\tilde{y}byc$, contrary to the assumption that $u$ is primitive (Lemma~\ref{lma:main} case  (iv)). 
	\begin{figure}
		\begin{subfigure}[t]{0.50\textwidth}
			\begin{center}
\resizebox{6cm}{!}{				
				\begin{tikzpicture}[scale= 0.2]
					\draw (3,0) to (26,0);
					\draw (0,0) -- (0.5,0);
					\draw[snake=zigzag,segment length = 5] (26,0) -- (28.5,0);
					\draw[->-=1] (28.5,0) -- (31,0);
					\draw[snake=zigzag,segment length = 5] (0.5,0) -- (3,0);
					\coordinate[label={$w$}] (wLabel) at (32,0);
					\draw[->-=1] (0,0) to (0, 12.5);
					\coordinate[label={\rotatebox{90}{$u$}}] (wLabel) at(-0.5, 10.5) ;
					\draw[->-=1] (0,15) to (0, 28);
					\coordinate[label={\rotatebox{90}{$v$}}] (wLabel) at(-0.5, 26) ;
					
					\coordinate (U) at (4,3);
					\coordinate (V) at (11,10);
					\coordinate (W) at (18,3);
					\coordinate (X) at (25,10);
					
					\coordinate (Y) at (13,16);
					
					\draw[dashed] ($(U) - (1,1)$) -- (U);
					\draw (U) -- (V) -- (W) -- (X);
					\draw[dashed] (X) -- ($(X) + (1,1)$);
					
					\draw[dashed] ($(Y) + (-7,7)$) -- ($(Y) + (-9,9)$);
					\draw ($(Y) + (-7,7)$) -- (Y) --  ($(Y) + (7,7)$);
					\draw ($(Y) + (7,7)$) -- ($(Y) + (9,9)$);
					
					\node (nU) at (U)  [element,label=above:{$U$}] {}; 
					\node (nV) at (V)  [element,label=above:{$V$}] {}; 
					\node (nW) at (W)  [element,label=above:{$W$}] {}; 
					\node (nX) at (X)  [element,label=above:{$X$}] {}; 
					\node (nY) at (Y)  [element,label=above:{$Y$}] {}; 
					
					\draw[dotted] (9,18) -- (9,0);
					\draw[dotted] (11,18) -- (11,0);
					\draw[dotted] (Y) -- (13,0);
					\draw[dotted] (15,18) -- (15,0);
					\draw[dotted] (17,20) -- (17,0);

					\draw[<->,dotted]  (9.2,15) to node[label=below:{$k$}] {} (10.8,15);			
					\draw[<->,dotted]  (11.2,15) to node[label=below:{$k$}] {} (12.8,15);
					\draw[<->,dotted]  (13.2,15) to node[label=below:{$k$}] {} (14.8,15);
					\draw[<->,dotted]  (15.2,15) to node[label=below:{$k$}] {} (16.8,15);

					\coordinate [label=$\tilde{y}$] (YW1) at (8,-4.25);			
					\coordinate [label=$c$] (w1) at (9,-2);
					\coordinate [label=$b$] (w2) at (11,-2);
					\coordinate [label=$y$] (w3) at (12,-4.25);
					\coordinate [label=$c$] (w4) at (13,-2);
					\coordinate [label=$\tilde{y}$] (w5) at (14,-4.25);
					\coordinate [label=$b$] (w6) at (15,-2);
					\coordinate [label=$c$] (w7) at (17,-2);
					\coordinate [label=$y$] (w8) at (18,-4.25);
					
					\draw[->] (12,-2) -- (12, 2);
					\draw[->] (14,-2) -- (14, 2);
					\draw[->] (8,-2) to[bend left=45] (10, 2);
					\draw[->] (18,-2) to[bend right=45] (16, 2);

					\filldraw[draw=green,fill=green,opacity= 0.3] (9,20) -- (9,0) -- (11,0) -- (11,18) -- cycle;
					\filldraw[draw=blue,fill=blue,opacity= 0.3] (11,18)  -- (11,0) -- (13,0)  -- (Y) -- cycle;
					\filldraw[draw=green,fill=green,opacity= 0.3] (Y)  -- (13,0) -- (15,0)  -- (15,18) -- cycle;
					\filldraw[draw=blue,fill=blue,opacity= 0.3] (15,18) -- (15,0) -- (17,0) -- (17,20) -- cycle;
				\end{tikzpicture}
}
			\end{center}
			\caption{Condition $k \leq \frac{1}{3}\ell$; for illustration,\\ $Y >V$.}
			\label{fig:close}
		\end{subfigure}
		\begin{subfigure}[t]{0.50\textwidth}
			\begin{center}
\resizebox{6cm}{!}{	
			\begin{tikzpicture}[scale= 0.2]
					\draw (3,0) to (26,0);
					\draw (0,0) -- (0.5,0);
					\draw[snake=zigzag,segment length = 5] (26,0) -- (28.5,0);
					\draw[->-=1] (28.5,0) -- (31,0);
					\draw[snake=zigzag,segment length = 5] (0.5,0) -- (3,0);
					\coordinate[label={$w$}] (wLabel) at (32,0);
					\draw[->-=1] (0,0) to (0, 12.5);
					\coordinate[label={\rotatebox{90}{$u$}}] (wLabel) at(-0.5, 10.5) ;
					\draw[->-=1] (0,15) to (0, 28);
					\coordinate[label={\rotatebox{90}{$v$}}] (wLabel) at(-0.5, 26) ;
					
					\coordinate (U) at (4,3);
					\coordinate (V) at (11,10);
					\coordinate (W) at (18,3);
					\coordinate (X) at (25,10);
					
					\coordinate (Y) at (8,16);
					\coordinate (Z) at (15.5,23.5);
					
					\draw (U) -- (V) -- (W) -- (X);
					\draw[dashed] (X) -- ($(X) + (1,1)$);
					
					\draw[dashed] ($(Y) + (-7,7)$) -- ($(Y) + (-9,9)$);
					\draw ($(Y) + (-7,7)$) -- (Y) --  (Z);
					\draw (Z) -- ($(Z) + (1,-1)$);
					
					\node (nU) at (U)  [element,label=above:{$U$}] {}; 
					\node (nV) at (V)  [element,label=above:{$V$}] {}; 
					\node (nW) at (W)  [element,label=above:{$W$}] {}; 
					\node (nX) at (X)  [element,label=above:{$X$}] {}; 
					\node (nY) at (Y)  [element,label=above:{$Y$}] {};
					\node (nZ) at (Z)  [element,label=above:{$Z$}] {};

					\draw[<->,dotted]  (5.2,15) to node[label=below:{$k$}] {} (7.8,15);			
					\draw[<->,dotted]  (8.2,15) to node[label=below:{$k$}] {} (10.8,15);
					\draw[<->,dotted]  (11.2,15) to node[label=below:{$k$}] {} (13.8,15);
					\draw[<->,dotted]  (14.2,15) to node[label=below:{$k$}] {} (16.8,15);

					\draw[dotted] (5,19) -- (5,0);			
					\draw[dotted] (Y) -- (8,0);
					\draw[dotted] (11,19) -- (11,0);
					\draw[dotted] (14,22) -- (14,0);
					\draw[dotted] (17,25) -- (17,0);

					\coordinate [label=$\tilde{y}$] (w8) at (4,-4.25);			
					\coordinate [label=$b$] (w7) at (5,-2);
					\coordinate [label=$c$] (w1) at (8,-2);
					\coordinate [label=$y$] (w3) at (9.5,-4.25);
					\coordinate [label=$b$] (w2) at (11,-2);
					\coordinate [label=$\tilde{y}$] (w5) at (12.5,-4.25);
					\coordinate [label=$c$] (w4) at (14,-2);
					\coordinate [label=$b$] (w6) at (17,-2);
					\coordinate [label=$y$] (w8) at (18,-4.25);
					
					\draw[->] (9.5,-2) -- (9.5, 2);
					\draw[->] (12.5,-2) -- (12.5, 2);
					\draw[->] (4,-2) to[bend left=45] (6, 2);
					\draw[->] (18,-2) to[bend right=45] (16, 2);
					\filldraw[draw=green,fill=green,opacity= 0.3]  (5,19) --  (5,0) -- (8,0) -- (Y) -- cycle;
					\filldraw[draw=blue,fill=blue,opacity= 0.3] (Y) -- (8,0) -- (11,0) -- (11,19) -- cycle;
					\filldraw[draw=green,fill=green,opacity= 0.3] (11,19) -- (11,0) --(14,0) -- (14,22) -- cycle;
					\filldraw[draw=blue,fill=blue,opacity= 0.3] (14,7)  -- (14,0) -- (17,0)  -- (17,4) -- cycle;
					
				\end{tikzpicture}
}
			\end{center}
			\caption{Condition $\frac{1}{3}\ell < k < \frac{1}{2}\ell$; for illustration, $Y <V$.}
			\label{fig:middling}
		\end{subfigure}
		\caption{The walk $g$ has a waypoint at $Y$ with $k = |V-Y| \geq 1$.}
		\label{fig:closeMiddling}
	\end{figure}
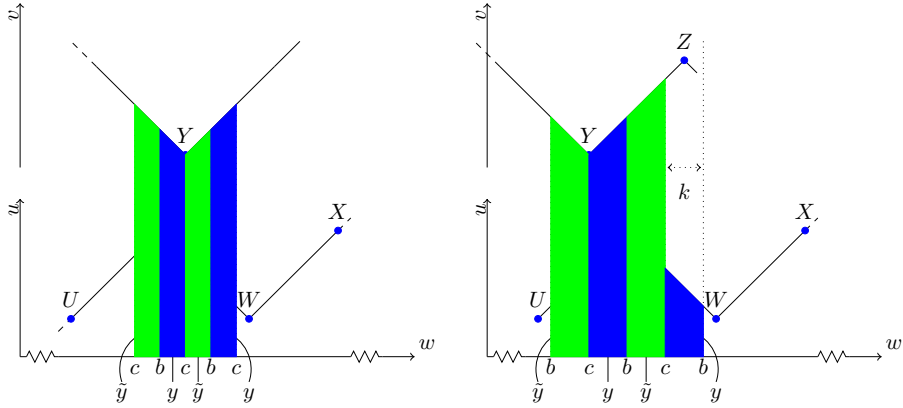

	\proofPara 
	Case (iii):  $\frac{1}{3}\ell < k < \frac{1}{2}\ell$. Again, in this case, we may have either $Y > V$ or $Y < V$.  This time (for variety) assume the latter; however, the
	reasoning in the former case is almost identical. 
	Thus, we have the situation depicted in Fig.~\ref{fig:middling}. 
	Let $w[V] = b$, $w[Y] = c$ and $w[Y{+}1,V{-}1] = y$. Since $Y$ is a waypoint on $g$, we see that
	$w[Y- k] = b$ and $w[Y{-} k{+}1, Y{-}1] = \tilde{y}$, whence $w[V{-} 2k,V] = w[Y{-} k,V] = b\tilde{y}cyb$. Since $V$ is a waypoint on $f$, we see that also
	$w[V, V{+}2k] = b\tilde{y}cyb$. Thus, $u$ contains the factor $b\tilde{y}cyb$ and $v$ contains the factor  $cyb\tilde{y}c$; moreover $w[Y, Y {+} 3k] = cyb\tilde{y}cyb$. 
	
	Now let $Z$ be the next waypoint on $g$ after $Y$. It is immediate that  $Z {-} Y < 3k$, since otherwise, $v$ contains the factor  $cyb\tilde{y}cyc$, contrary to the assumption that $v$ is primitive (Lemma~\ref{lma:main}  case (iv)). 
	We consider three possibilities for the point $Z$, depending on  where, exactly, $Z$ is positioned in $[V{+} k, V{+}2k] = [Y{+} 2k, Y{+}3k]$. The three possibilities are 
	indicated in Fig.~\ref{fig:threePossibilities}, which shows the detail of Fig.~\ref{fig:middling} in that interval.
	Suppose (a) that $V{+}k < Z < V{+} \frac{3}{2}k$. Then, by inspection of  Fig.~\ref{fig:middlingA}, $y$ must be of the form $xd\tilde{x}cz$ for some letter $d$ and strings $x$ and $z$.
	But we have already argued that $u$ contains  the factor 
	\begin{equation*}
		b\tilde{y}cyb= b(xd\tilde{x}cz)^{-1}c(xd\tilde{x}cz)b =    b(\tilde{z}cxd\tilde{x})c(xd\tilde{x}cz)b
	\end{equation*}
	and hence the factor $cxd\tilde{x}cxd$ contrary to the assumption that $u$ is primitive (Lemma~\ref{lma:main} case  (iv)).
	Suppose (b) that $Z = V{+} \frac{3}{2}k$. Then, by inspection of  Fig.~\ref{fig:middlingB}, $y$ must be of the form $xd\tilde{x}$ for some letter $d$ and string $x$, and 
	furthermore, $b = c$.
	But in that case  $u$ contains  the factor 
	\begin{equation*}
		b\tilde{y}cyb= c(xd\tilde{x})^{-1}c(xd\tilde{x})c =    c(xd\tilde{x})c(xd\tilde{x})c
	\end{equation*}
	and hence the factor $cxd\tilde{x}cxd$ again. Suppose (c) that $V{+} \frac{3}{2}k < Z < V{+} 2k$. Then 
	by inspection of  Fig.~\ref{fig:middlingC}, $y$ must be of the form $zbxd\tilde{x}$ for some letter $d$ and strings $x$ and $z$.
	But we have already argued that $v$ contains  the factor
	\begin{equation*}
		cyb\tilde{y}c = c(zbxd\tilde{x})b(zbxd\tilde{x})^{-1}c = c(zbxd\tilde{x})b(xd\tilde{x}b\tilde{z})c 
	\end{equation*}
	and hence the factor $bxd\tilde{x}bxd$, again contrary to the assumption that $u$ is primitive.
	This eliminates all possibilities for the position of $Z$, and thus yields the desired contradiction.
	
	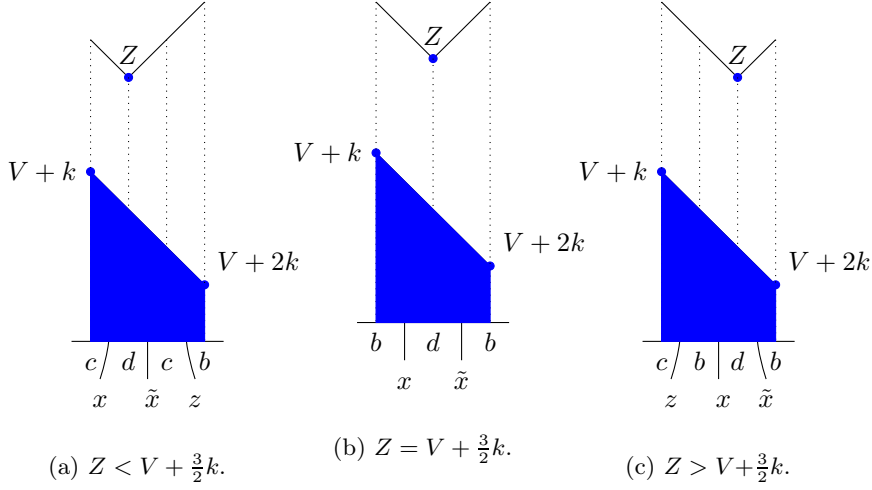
\begin{figure}
		\begin{subfigure}[t]{0.30\textwidth}
			\begin{center}
				\begin{tikzpicture}[scale= 0.25]
					
					\draw (2,0) -- (10,0);
					
					\coordinate (Vl) at (3,9);
					\coordinate (V2l) at (9,3);
					\coordinate (Z) at (5,14);
					
					\draw (Vl) -- (V2l);
					\draw ($ (Z) + (-2,2) $) -- (Z) -- ($ (Z) + (4,4) $);
					
					\node (Vln) at (Vl)  [element,label=left:{$V+k$}] {}; 
					\node (V2ln) at (V2l)  [element,label=above right:{$V+2k$}] {}; 
					\node (Zn) at (Z)  [element,label=above:{$Z$}] {}; 
					
					\draw[dotted] (3,16) -- (3,0);
					\draw[dotted] (5,14) -- (5,0);
					\draw[dotted] (7,16) -- (7,0);
					\draw[dotted] (9,18) -- (9,0);
					
					\coordinate[label={$c$}] (Vlw) at (3,-2);
					\coordinate[label={$d$}] (Zw) at (5,-2);
					\coordinate[label={$c$}] (Vlm) at (7,-2);
					\coordinate[label={$b$}] (V2lw) at (9,-2);
					
					\coordinate[label={$x$}] (Vly) at (3.5,-4);
					\coordinate[label={$\tilde{x}$}] (Vym) at (6.25,-4);
					\coordinate[label={$z$}] (Vly) at (8.5,-4);
					
					\draw[->] (3.5,-2) to [bend right=10] (4,2);
					\draw[->] (6,-2) -- (6,2);
					\draw[->] (8.5,-2) to [bend left=10] (8,2);
					
					\filldraw[draw=blue,fill=blue,opacity= 0.3] (3,9) -- (3,0) -- (9,0) -- (9,3) -- cycle;
				\end{tikzpicture}
			\end{center}
			\caption{$Z < V+ \frac{3}{2}k$.}
			\label{fig:middlingA}
		\end{subfigure}
		\begin{subfigure}[t]{0.3\textwidth}
			\begin{center}
				\begin{tikzpicture}[scale= 0.25]
					
					\draw (2,0) -- (10,0);
					
					\coordinate (Vl) at (3,9);
					\coordinate (V2l) at (9,3);
					\coordinate (Z) at (6,14);
					
					\draw (Vl) -- (V2l);
					\draw ($ (Z) + (-3,3) $) -- (Z) -- ($ (Z) + (3,3) $);
					
					\node (Vln) at (Vl)  [element,label=left:{$V+k$}] {}; 
					\node (V2ln) at (V2l)  [element,label=above right:{$V+2k$}] {}; 
					\node (Zn) at (Z)  [element,label=above:{$Z$}] {}; 
					
					\draw[dotted] (3,17) -- (3,0);
					\draw[dotted] (6,14) -- (6,0);
					\draw[dotted] (9,17) -- (9,0);
					
					\coordinate[label={$b$}] (Vlw) at (3,-2);
					\coordinate[label={$d$}] (Zw) at (6,-2);
					\coordinate[label={$b$}] (V2lw) at (9,-2);
					
					\coordinate[label={$x$}] (Vly) at (4.5,-4);
					\coordinate[label={$\tilde{x}$}] (Vym) at (7.5,-4);
					
					\draw[->] (4.5,-2) to (4.5,2);
					\draw[->] (7.5,-2) to (7.5,2);
					
					\filldraw[draw=blue,fill=blue,opacity= 0.3] (3,9) -- (3,0) -- (9,0) -- (9,3) -- cycle;
				\end{tikzpicture}
			\end{center}
			\caption{$Z = V+ \frac{3}{2}k$.}
			\label{fig:middlingB}
		\end{subfigure}
		\begin{subfigure}[t]{0.30\textwidth}
			\begin{center}
				\begin{tikzpicture}[scale= 0.25]
					
					\draw (2,0) -- (10,0);
					
					\coordinate (Vl) at (3,9);
					\coordinate (V2l) at (9,3);
					\coordinate (Z) at (7,14);
					
					\draw (Vl) -- (V2l);
					\draw ($ (Z) + (-4,4) $) -- (Z) -- ($ (Z) + (2,2) $);
					
					\node (Vln) at (Vl)  [element,label=left:{$V+k$}] {}; 
					\node (V2ln) at (V2l)  [element,label=above right:{$V+2k$}] {}; 
					\node (Zn) at (Z)  [element,label=above:{$Z$}] {}; 
					
					\draw[dotted] (3,18) -- (3,0);
					\draw[dotted] (5,16) -- (5,0);
					\draw[dotted] (7,14) -- (7,0);
					\draw[dotted] (9,16) -- (9,0);
					
					\coordinate[label={$c$}] (Vlw) at (3,-2);
					\coordinate[label={$b$}] (Zw) at (5,-2);
					\coordinate[label={$d$}] (Vlm) at (7,-2);
					\coordinate[label={$b$}] (V2lw) at (9,-2);
					
					\coordinate[label={$z$}] (Vly) at (3.5,-4);
					\coordinate[label={$x$}] (Vym) at (6.25,-4); 
					\coordinate[label={$\tilde{x}$}] (Vly) at (8.5,-4);
					
					\draw[->] (3.5,-2) to [bend right=10] (4,2);
					\draw[->] (6,-2) -- (6,2);
					\draw[->] (8.5,-2) to [bend left=10] (8,2);
					
					\filldraw[draw=blue,fill=blue,opacity= 0.3] (3,9) -- (3,0) -- (9,0) -- (9,3) -- cycle;
				\end{tikzpicture}
			\end{center}
			\caption{$Z > V{+} \frac{3}{2}k$.}
			\label{fig:middlingC}
		\end{subfigure}
		\caption{The location of $Z$ with respect to $V{+} \frac{3}{2}k$ in Case (iii).}
		\label{fig:threePossibilities}
	\end{figure}
	
	\proofPara 
	Case (iv):  $k > \frac{1}{2}\ell$. Since, by assumption, $Y$ is not closer to $W$ than it is to $V$, we see that $Y$ cannot lie in the interval $[V,W]$,
	i.e., $Y < V$. Indeed, $g$ cannot have any waypoint on the interval
	$[V,W {+} k{-}1]$, so that we have the situation depicted in Fig.~\ref{fig:UYbig}. Write $w[U]= a$, $w[Y]=c$ and $w[U{+}1,Y-1]=y$. In addition, let us write $h = \ell -k < \frac{1}{2}\ell$. Since $Y$ is a waypoint on $g$, $w[Y,Y{+}h] = c\tilde{y}a$,
	whence $w[U,U{+}2h] = ayc\tilde{y}a$. Using the fact that $V$ is a waypoint on $f$ with $U$ and $W$ symmetrically positioned with respect to $V$, we see that, also
	$w[W{-}2h,W] = ayc\tilde{y}a$. On the other hand, $W$ is also a waypoint on $f$, whence $w[W,Y{+}h] = ayc$. Thus $w[W{-}2h,W{+} h] = ayc\tilde{y}ayc$. 
	But we have already observed that $g$ has no waypoint on the interval $[V,W{+} k-1] \supset [V,W{+} h {-}1]$, whence $v$ has a factor $ayc\tilde{y}ayc$, contrary to the assumption that $v$ is primitive (Lemma~\ref{lma:main}  case (iv)).
	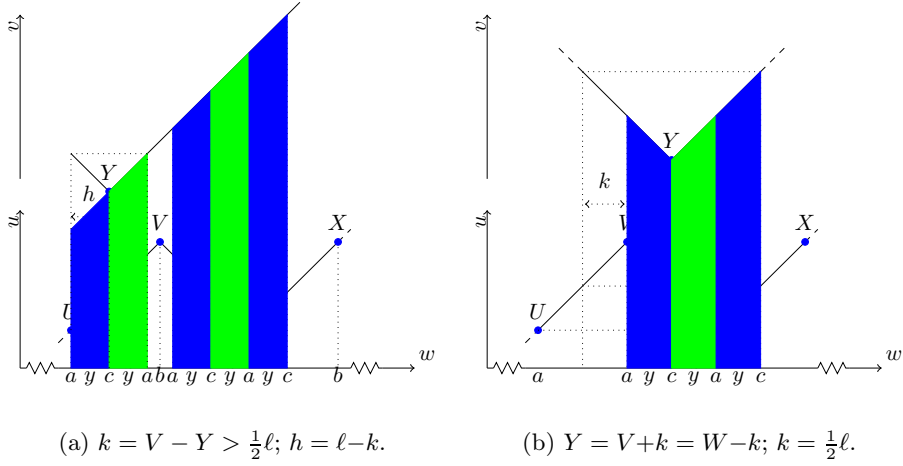
\begin{figure}
		\begin{subfigure}[t]{0.50\textwidth}
			\begin{center}
\resizebox{6cm}{!}{	
				\begin{tikzpicture}[scale= 0.2]
					\draw (3,0) to (26,0);
					\draw (0,0) -- (0.5,0);
					\draw[snake=zigzag,segment length = 5] (26,0) -- (28.5,0);
					\draw[->-=1] (28.5,0) -- (31,0);
					\draw[snake=zigzag,segment length = 5] (0.5,0) -- (3,0);
					\coordinate[label={$w$}] (wLabel) at (32,0);

					\draw[->-=1] (0,0) to (0, 12.5);
					\coordinate[label={\rotatebox{90}{$u$}}] (wLabel) at(-0.5, 10.5) ;
					
					\coordinate (U) at (4,3);
					\coordinate (V) at (11,10);
					\coordinate (W) at (18,3);
					\coordinate (X) at (25,10);
					
					\coordinate [label=$a$] (Uw1) at (4,-1.75);
					\coordinate [label=$y$] (Uw2) at (5.5,-2);
					\coordinate [label=$c$] (Uw3) at (7,-1.75);
					\coordinate [label=$\tilde{y}$] (Uw2) at (8.5,-2);
					\coordinate [label=$a$] (Vw4) at (10,-1.75);
					
					\coordinate [label=$b$] (Vw4) at (11,-1.75);
					
					\coordinate [label=$a$] (Uw1) at (12,-1.75);
					\coordinate [label=$y$] (Uw2) at (13.5,-2);
					\coordinate [label=$c$] (Uw3) at (15,-1.75);
					\coordinate [label=$\tilde{y}$] (Uw2) at (16.5,-2);
					\coordinate [label=$a$] (Vw4) at (18,-1.75);
					\coordinate [label=$y$] (Uw2) at (19.5,-2);
					\coordinate [label=$c$] (Uw3) at (21,-1.75);
					
					\coordinate [label=$b$] (Vw5) at (25,-1.75);
					
					\draw[<->,dotted] (4,12) to node [label=above:{$h$}]{} (6.8,12);
					\draw[<->,dotted] (7.2,12) to node [label=above:{$h$}]{} (10,12);
					\draw[<->,dotted] (12.2,15) to node [label=above:{$h$}]{} (14.8,15);
					\draw[<->,dotted] (15.2,15) to node [label=above:{$h$}]{} (17.8,15);
					\draw[<->,dotted] (18.2,15) to node [label=above:{$h$}]{} (21,15);
					\draw[dotted] (12,0) to (12,19);
					\draw[dotted] (18,0) to (18,25);
					\draw[dotted] (15,0) to (15,22);
					\draw[dotted] (21,0) to (21,28);

					\draw[dashed] ($(U) - (1,1)$) -- (U);
					\draw (U) -- (V) -- (W) -- (X);
					\draw[dashed] (X) -- ($(X) + (1,1)$);
					
					
					\draw[dotted]  (U) to ($ (U) - (0,3) $); 
					\draw[dotted]  (V) to ($ (V) - (0,10) $);
					\draw[dotted]  ($(W) + (0,14) $) to ($ (W) - (0,3) $);  
					\draw[dotted]  (X) to ($ (X) - (0,10) $);

					\node (nU) at (U)  [element,label=above:{$U$}] {}; 
					\node (nV) at (V)  [element,label=above:{$V$}] {}; 
					\node (nW) at (W)  [element,label=above:{$W$}] {}; 
					\node (nX) at (X)  [element,label=above:{$X$}] {}; 
					
					
					\draw[->-=1] (0,15) to (0, 28);
					\coordinate[label={\rotatebox{90}{$v$}}] (wLabel) at(-0.5, 26) ;
					
					\coordinate (T) at (4, 17);
					\coordinate (Y) at (7, 14);
					\coordinate (Z) at (21,28);
					
					\draw ($ (U) + (0,14)  $) -- (Y) -- ($ (W) + (4,26)  $);
					\node (nW) at (Y)  [element,label=above:{$Y$}] {}; 
					
					\filldraw[draw=blue,fill=blue,opacity= 0.3] (Y) -- ($ (Y) - (3,3) $) -- ($ (U) - (0,3) $) -- ($ (U) + (3,-3) $)  -- cycle;
					\filldraw[draw=green,fill=green,opacity= 0.3] (Y) -- ($ (Y) + (3,3) $) -- ($ (U) + (6,-3) $) -- ($ (U) + (3,-3) $)  -- cycle;
					
					\filldraw[draw=blue,fill=blue,opacity= 0.3] ($ (W) + (-6,-3) $)  -- ($ (W) + (-6,16) $) -- ($ (W) + (-3,19) $)  -- ($ (W) + (-3,-3) $) -- cycle;
					\filldraw[draw=green,fill=green,opacity= 0.3] ($ (W) + (-3,-3) $)  -- ($ (W) + (-3,19) $) -- ($ (W) + (0,22) $)  -- ($ (W) + (0,-3) $) -- cycle;
					\filldraw[draw=blue,fill=blue,opacity= 0.3]  ($ (W) + (0,-3) $)  -- ($ (W) + (0,22) $) -- ($ (W) + (3,25) $)  -- ($ (W) + (3,-3) $) -- cycle;
					
					\draw[dotted]  (T) to (U);
					\draw[dotted]  (Y) to ($ (Y) - (0,14) $);
					\draw[dotted]  ($ (T) + (6,0)  $) to ($ (U) + (6,-3)  $);
					
					\draw[dotted]  (T) to ($ (T) + (6,0) $);
				\end{tikzpicture}
}
			\end{center}
			\caption{$k= V-Y > \frac{1}{2}\ell$; $h=\ell {-}k$.}
			\label{fig:UYbig}
		\end{subfigure}
		\begin{subfigure}[t]{0.50\textwidth}
			\begin{center}
\resizebox{6cm}{!}{	
					\begin{tikzpicture}[scale= 0.2]
					\draw (3,0) to (26,0);
					\draw (0,0) -- (0.5,0);
					\draw[snake=zigzag,segment length = 5] (26,0) -- (28.5,0);
					\draw[->-=1] (28.5,0) -- (31,0);
					\draw[snake=zigzag,segment length = 5] (0.5,0) -- (3,0);
					\coordinate[label={$w$}] (wLabel) at (32,0);
					\draw[->-=1] (0,0) to (0, 12.5);
					\coordinate[label={\rotatebox{90}{$u$}}] (wLabel) at(-0.5, 10.5) ;
					\draw[->-=1] (0,15) to (0, 28);
					\coordinate[label={\rotatebox{90}{$v$}}] (wLabel) at(-0.5, 26) ;
					
					\coordinate (U) at (4,3);
					\coordinate (V) at (11,10);
					\coordinate (W) at (18,3);
					\coordinate (X) at (25,10);
					
					\coordinate (Y) at (14.5,16.5);
					
					\draw[dashed] ($(U) - (1,1)$) -- (U);
					\draw (U) -- (V) -- (W) -- (X);
					\draw[dashed] (X) -- ($(X) + (1,1)$);
					
					\draw[dashed] ($(Y) + (-7,7)$) -- ($(Y) + (-9,9)$);
					\draw ($(Y) + (-7,7)$) -- (Y) --  ($(Y) + (7,7)$);
					\draw[dashed] ($(Y) + (7,7)$) -- ($(Y) + (9,9)$);
					
					\node (nU) at (U)  [element,label=above:{$U$}] {}; 
					\node (nV) at (V)  [element,label=above:{$V$}] {}; 
					\node (nW) at (W)  [element,label=above:{$W$}] {}; 
					\node (nX) at (X)  [element,label=above:{$X$}] {}; 
					\node (nY) at (Y)  [element,label=above:{$Y$}] {}; 
					
					\draw[dotted] (11,20) -- (11,0);
					\draw[dotted] (Y) -- (14.5,0);
					\draw[dotted] (18,20) -- (18,0);
					\draw[dotted] (21.5,23.5) -- (21.5,0);
					\draw[dotted] (U) to (W);
					\draw[dotted] ($ (U) + (3.5,3.5) $) to ($ (W) + (3.5,3.5) $);
					\draw[dotted] ($ (U) + (3.5,20.5) $) to ($ (W) + (3.5,20.5) $);
					\draw[dotted] (7.5,23.5) -- (7.5,0);
					
					\draw[<->,dotted] (7.7,13) to node [label=above:{$k$}]{} (10.8,13);
					\draw[<->,dotted] (11.2,13) to node [label=above:{$k$}]{} (14.3,13);
					\draw[<->,dotted] (14.7,13) to node [label=above:{$k$}]{} (17.8,13);
					\draw[<->,dotted] (18.2,13) to node [label=above:{$k$}]{} (21.3,13);
					
					\coordinate [label=$a$] (U1) at (4,-1.75);
					\coordinate [label=$a$] (V1) at (11,-1.75);
					\coordinate [label=$y$] (VY1) at (12.75,-2);
					\coordinate [label=$c$] (Y1) at (14.5,-1.75);
					\coordinate [label=$\tilde{y}$] (YW1) at (16.25,-2);
					\coordinate [label=$a$] (W1) at (18,-1.75);
					\coordinate [label=$y$] (WWW1) at (19.75,-2);
					\coordinate [label=$c$] (WW1) at (21.5,-1.75);

					\filldraw[draw=blue,fill=blue,opacity= 0.3] (11,20)  -- (11,0) -- (14.5,0)  -- (Y) -- cycle;
					\filldraw[draw=green,fill=green,opacity= 0.3] (Y) -- (14.5,0) -- (18,0) -- (18,20) -- cycle;
					\filldraw[draw=blue,fill=blue,opacity= 0.3] (18,20)  -- (18,0) -- (21.5,0)  -- (21.5,23.5) -- cycle;				
				\end{tikzpicture}
}
			\end{center}
			\caption{$Y = V {+} k = W {-} k$; $k= \frac{1}{2}\ell$.}
			\label{fig:UYmidway}
		\end{subfigure}
		\caption{The walk $g$ has a waypoint at $Y$.}
		\label{fig:UYbigAndMidway}
	\end{figure}

	\proofPara 
	Case (v):  $k = \frac{1}{2}\ell$. We first claim that $g$ has a waypoint in the interval $[V{+}1, W-1]$. For otherwise, we have $Y < V$, and the same
	situation
	as depicted in   Fig.~\ref{fig:UYbig} arises, except that the unshaded region around $V$ disappears, with $a = b$. But we have already argued that,
	in that case, $v$ has a factor $ayc\tilde{y}ayc$, contrary to the assumption that $v$ is primitive. 
	
	Thus, we may assume that $ Y = V {+} k = W {-} k$. 
	By the minimality of the leg $[V,W]$, we know that $g$ has no other waypoints in the interval
	$[V {-} k{+}1, W {+} k {-}1]$, and we have the situation depicted in Fig.~\ref{fig:UYmidway}.
	Observe that $f(V {-} k) = f(W+{k})$ and $g(V {-} k) = g(W+{k})$.
	By inspection, we see that the stroll $f'= f /  [V{-}k, W{+}k]$ is in fact surjective, i.e.~a walk.
	If, in addition
	$g([1,V{-} k]) \cup g([W{+}k,n])  = [1,|v|]$, we can define
	$g'= g /  [V{-}k, W{+}k]$, and $w' =  w[1,V{-}k] \cdot w[W{+}k{+}1,|u|]$, 
	so that
	$f'$ generates $w'$ from $u$ and $g'$ generates $w'$ from $v$, contradicting the assumption that $w$ is a shortest counterexample. Hence, $g$
	takes values
	on the interval $[V{-}k{+}1, W{+}k-1] = [Y{-}\ell{+}1, Y{+}\ell-1]$ not taken outside this interval, whence $g(Y)$ is a terminal point of $v$. Since we have drawn $Y$ as a trough, $g(Y) = 1$, but the same reasoning applies,
	with the obvious changes, if $Y$ is a peak.
	Write $w[U] = a$, $w[Y] = c$ and  $w[V{+}1,Y-1] = y$. From the fact that $V$ is a waypoint on $f$, 
	we have $w[W] = a$, and from the fact that 
	$Y$ is a waypoint on $g$, we have
	$w[V{+}1,W-1]= \tilde{y}$ and $w[V] = a$, whence $w[V,W] = ayc\tilde{y}a$. Using the 
	fact that $W$ is a waypoint on $f$ again, 
	$w[W{+}1,W{+}k-1]= y$ and 
	$w[W{+}k]= c$, whence $w[Y{+}1,Y{+}2k]= c\tilde{y}ayc$.
	Recalling that $g(Y)= 1$, 
	we see that $v$ has a prefix $c\tilde{y}ayc$ starting at $Y$, contrary to the supposition that $v$ is primitive (Lemma~\ref{lma:main} case  (i)).

	\proofPara
	This deals with all cases in which the minimal leg $[V,W]$ is \textit{internal}. We turn finally to the few remaining cases where it is \textit{terminal}. Without loss of generality, we may assume that 
	we are dealing with an initial leg, i.e. $V = 1$; hence $W$ is an internal waypoint. Further, by replacing $v$ with its reversal if necessary, we may assume that the initial leg of $g$ is ascending.  As before, let $X = W {+}\ell$.  By the minimality of $[V,W]$, we have $X \leq n$. 
	Again by the minimality of $[V,W]$, we have $f([W,n]) = [1,|u|]$.
	We claim that $g([W,n]) \neq [1,|v|]$. For otherwise, defining $f'= f /  [1,W]$,   $g'= g /  [1,W]$ and $w' = w[W,n]$, we see that $f'$ generates $w'$ from $u$
	and $g'$ generates $w'$ from $v$, contrary to the assumption that $w$ is the shortest counterexample. Thus, $g$ achieves values on the interval $[1,W-1]$ not achieved outside this interval, whence, since this is an ascending leg, $g(1) = 1$. It also follows that $g$ does not have a waypoint at $W$, since, otherwise, by the minimality of the leg $[V,W]$, the interval $[W,X]$ is
	included in a leg of $g$, whence $g([W,n]) = [1,|v|]$, which we have shown is not true. Now write $w[1] = a$,  $w[W] = b$ and $w[2,W-1] = x$. Since $W$ is a waypoint on $f$, we have
	$w[X] = a$ and $w[W{+}1,X-1] = \tilde{x}$. Observe that $g$ must have a waypoint, say $Y$, in the interval $[W{+}1,X-1]$, since otherwise $v$ has a prefix $axb\tilde{x}a$,
	contrary to the assumption that $v$ is primitive (Lemma~\ref{lma:main} case  (i)). By the minimality of the leg $[V,W]$, $Y$ is unique. Write $c = w[Y]$ and
	$j = Y-W$. This time, we have just two cases.
	
	\proofPara
	Case (i): $j \leq \frac{1}{2}\ell$. The situation is shown in Fig.~\ref{fig:initSmallJ}. 
	Write $y = w[W{+}1, Y-1]$, so that $w[W, Y] = w[Y-j, Y] = byc$.
	Since $Y$ is a waypoint on $g$, we have $w[Y, Y{+} j] = c\tilde{y}b$, and hence 
	$w[W, W{+} 2j] = byc\tilde{y}b$. Therefore, since
	$W$ is a waypoint on $f$, we have $w[W-2j, W] = byc\tilde{y}b$, whence $w[W-2j, Y] = byc\tilde{y}byc$. But $Y$ is the first internal waypoint  on $g$, whence $v$ has a factor
	$byc\tilde{y}byc$, contrary to the assumption that $v$ is primitive (Lemma~\ref{lma:main} case  (iv)).  
	
	\proofPara 
	Case (ii): $j > \frac{1}{2}\ell$. The situation is shown in Fig.~\ref{fig:initLargeJ}. Setting $h = \ell {-} j$, we see that $X = Y{+}h$ and $h < j$. Since $Y$ is a waypoint on $g$, we have $w[Y{-}h] = w[Y{+}h] = w[X] = a$.
	Write $y= w[Y{-}h{+}1,Y]$,
	so that $w[X{-}2h,X] =w[Y{-}h{+}1,X] = ayc\tilde{y}a$.
	Therefore, since
	$W$ is a waypoint on $f$, and the points $W=1$ and $X$ are symmetrically positioned about $W$, we see that 
	also $w[1,2h] = a\tilde{y}cya$. 
	But $g$ certainly has no internal waypoint on this interval, whence $v$ has a prefix
	$a\tilde{y}cya$, again contrary to the assumption that $v$ is primitive (Lemma~\ref{lma:main} case  (ii)).  
	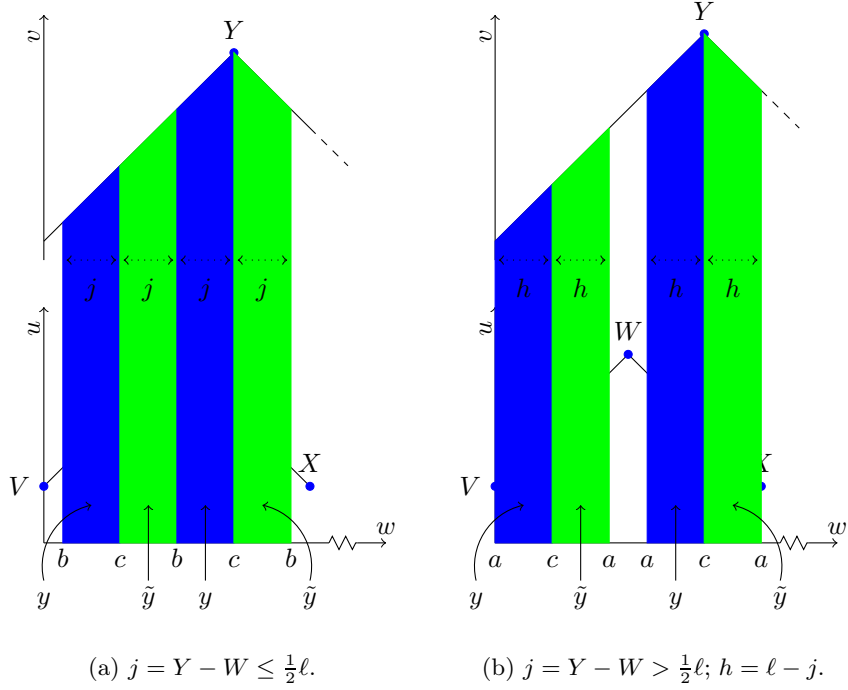
\begin{figure}
		\begin{subfigure}[t]{0.48\textwidth}
			\begin{center}
				\begin{tikzpicture}[scale= 0.25]
					\draw (0,0) -- (15,0);
					\draw[snake=zigzag,segment length = 5] (15,0) -- (16.5,0);
					\draw[->-=1] (16.5,0) -- (18,0);
					\coordinate[label={$w$}] (wLabel) at (18,0);
					
					\draw[->-=1] (0,0) to (0, 12.5);
					\coordinate[label={\rotatebox{90}{$u$}}] (wLabel) at(-0.5, 10.5) ;
					\draw[->-=1] (0,15) to (0, 28);
					\coordinate[label={\rotatebox{90}{$v$}}] (wLabel) at(-0.5, 26) ;

					\coordinate (V) at (0,3);
					\coordinate (Y) at (10,26);
					\coordinate (W) at (7,10);
					\coordinate (X) at (14,3);
					
					\draw (V) -- (W) -- (X);
					\draw ($ (Y) - (10, 10) $) -- (Y) -- ($ (Y) + (4, -4) $);
					\draw[dashed] ($ (Y) + (4, -4) $) -- ($ (Y) + (6, -6) $);

					\node (nV) at (V)  [element,label=left:{$V$}] {}; 
					\node (nW) at (W)  [element,label=above:{$W$}] {}; 
					\node (nX) at (X)  [element,label=above:{$X$}] {}; 			
					\node (nY) at (Y)  [element,label=above:{$Y$}] {};

					\filldraw[draw=blue,fill=blue,opacity= 0.3]  (4,0) --  (4,20) -- (1,17) -- (1,0) -- cycle;
					\filldraw[draw=green,fill=green,opacity= 0.3] (7,0) -- (7,23) -- (4,20) -- (4,0) -- cycle;
					\filldraw[draw=blue,fill=blue,opacity= 0.3] (7,0) -- (7,23) --(10,26) -- (10,0) -- cycle;
					\filldraw[draw=green,fill=green,opacity= 0.3] (10,0)  -- (10,26) -- (13,23)  -- (13,0) -- cycle;

					\draw[<->,dotted]  (1.2,15) to node[label=below:{$j$}] {} (3.8,15);			
					\draw[<->,dotted]  (4.2,15) to node[label=below:{$j$}] {} (6.8,15);
					\draw[<->,dotted]  (7.2,15) to node[label=below:{$j$}] {} (9.8,15);
					\draw[<->,dotted]  (10.2,15) to node[label=below:{$j$}] {} (12.8,15);

					\coordinate [label=$y$] (YW1) at (0,-4);			
					\coordinate [label=$b$] (w1) at (1,-1.75);
					\coordinate [label=$c$] (w2) at (4,-1.75);
					\coordinate [label=$\tilde{y}$] (w3) at (5.5,-4);
					\coordinate [label=$b$] (w4) at (7,-1.75);
					\coordinate [label=$y$] (w5) at (8.5,-4);
					\coordinate [label=$c$] (w6) at (10,-1.75);
					\coordinate [label=$b$] (w7) at (13,-1.75);
					\coordinate [label=$\tilde{y}$] (w8) at (14,-4);
					
					\draw[->] (5.5,-2) -- (5.5, 2);
					\draw[->] (8.5,-2) -- (8.5, 2);
					\draw[->] (0,-2) to[bend left=45] (2.5, 2);
					\draw[->] (14,-2) to[bend right=45] (11.5, 2);
				\end{tikzpicture}
			\end{center}
			\caption{$j= Y-W \leq \frac{1}{2}\ell$.}
			\label{fig:initSmallJ}
		\end{subfigure}
		\begin{subfigure}[t]{0.48\textwidth}
			\begin{center}
				\begin{tikzpicture}[scale= 0.25]
					\draw (0,0) -- (15,0);
					\draw[snake=zigzag,segment length = 5] (15,0) -- (16.5,0);
					\draw[->-=1] (16.5,0) -- (18,0);
					\coordinate[label={$w$}] (wLabel) at (18,0);
					
					\draw[->-=1] (0,0) to (0, 12.5);
					\coordinate[label={\rotatebox{90}{$u$}}] (wLabel) at(-0.5, 10.5) ;
					\draw[->-=1] (0,15) to (0, 28);
					\coordinate[label={\rotatebox{90}{$v$}}] (wLabel) at(-0.5, 26) ;

					\coordinate (V) at (0,3);
					\coordinate (Y) at (11,27);
					\coordinate (W) at (7,10);
					\coordinate (X) at (14,3);
					
					\draw (V) -- (W) -- (X);
					\draw ($ (Y) - (11, 11) $) -- (Y) -- ($ (Y) + (3, -3) $);
					\draw[dashed] ($ (Y) + (3, -3) $) -- ($ (Y) + (5, -5) $);

					\node (nV) at (V)  [element,label=left:{$V$}] {}; 
					\node (nW) at (W)  [element,label=above:{$W$}] {}; 
					\node (nX) at (X)  [element,label=above:{$X$}] {}; 			
					\node (nY) at (Y)  [element,label=above:{$Y$}] {};

					\filldraw[draw=blue,fill=blue,opacity= 0.3]  (3,0) --  (3,19) -- (0,16) -- (0,0) -- cycle;
					\filldraw[draw=green,fill=green,opacity= 0.3] (3,0) -- (3,19) -- (6,22) -- (6,0) -- cycle;
					
					\filldraw[draw=blue,fill=blue,opacity= 0.3] (8,0) -- (8,24) --(11,27) -- (11,0) -- cycle;
					\filldraw[draw=green,fill=green,opacity= 0.3] (11,0)  -- (11,27) -- (14,24)  -- (14,0) -- cycle;

					\draw[<->,dotted]  (0.2,15) to node[label=below:{$h$}] {} (2.8,15);			
					\draw[<->,dotted]  (3.2,15) to node[label=below:{$h$}] {} (5.8,15);
					\draw[<->,dotted]  (8.2,15) to node[label=below:{$h$}] {} (10.8,15);
					\draw[<->,dotted]  (11.2,15) to node[label=below:{$h$}] {} (13.8,15);

					\coordinate [label=$y$] (YW1) at (-1,-4);			
					\coordinate [label=$a$] (w1) at (0,-1.75);
					\coordinate [label=$c$] (w2) at (3,-1.75);
					\coordinate [label=$a$] (w1) at (6,-1.75);
					\coordinate [label=$\tilde{y}$] (w3) at (4.5,-4);
					\coordinate [label=$a$] (w4) at (8,-1.75);
					\coordinate [label=$y$] (w5) at (9.5,-4);
					\coordinate [label=$c$] (w6) at (11,-1.75);
					\coordinate [label=$a$] (w7) at (14,-1.75);
					\coordinate [label=$\tilde{y}$] (w8) at (15,-4);
					
					\draw[->] (4.5,-2) -- (4.5, 2);
					\draw[->] (9.5,-2) -- (9.5, 2);
					\draw[->] (-1,-2) to[bend left=45] (1.5, 2);
					\draw[->] (15,-2) to[bend right=45] (12.5, 2);
				\end{tikzpicture}
			\end{center}
			\caption{$j= Y-W > \frac{1}{2}\ell$; $h = \ell-j$.}
			\label{fig:initLargeJ}
		\end{subfigure}	
		\caption{The walk $g$ has its first internal waypoint at $Y \in [W{+}1, X-1]$.}
		\label{fig:init}
	\end{figure}
	
	\proofPara
	Thus, the assumption that $w$ is a shortest word for which there exist primitive generators $u$ and $v$ such that 
	$u \neq v$ and $u \neq \tilde{v}$ leads in all cases to a contradiction. This completes the proof that no such $w$ exists.
\end{proof}

\section{Uniqueness of walks}
\label{sec:walks}	
In this short  section, we prove Theorem~\ref{theo:noble}.
\newtheorem*{Restatetheo:noble}{Theorem~\ref{theo:noble}}
\begin{Restatetheo:noble}
	Let $u$ be a word. Then $u$ is perfect if and only if it contains no non-trivial palindrome as a factor. 
\end{Restatetheo:noble}
\begin{proof}
For the only-if direction, suppose that
$u$ contains a non-trivial palindrome. If that palindrome is odd, so that $u$ has the form $xayb\ti{y}az$, then the 
word $xayb\ti{y}ayb\ti{y}az$ is generated via the distinct walks $f$
and $g$ illustrated in Fig.~\ref{fig:notUnique}. If the contained palindrome is even, 
so that $u$ has the form $xaaz$, then the word $xaaaz$ is generated via distinct walks, one of which pauses for one step on the first $a$, and the other on the
second.
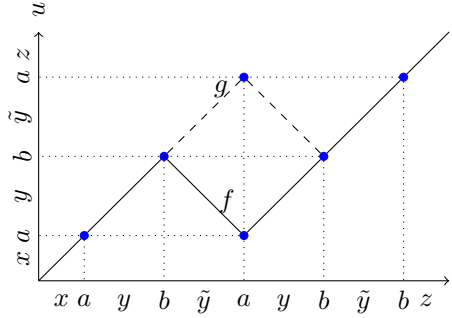
\begin{figure}
	\begin{center}
		\begin{tikzpicture}[scale= 0.15]
			
			\draw[->] (0,0) -- (36,0);
			
			\draw[->] (0,0) to (0, 22);
			\coordinate[label={\rotatebox{90}{$u$}}] (wLabel) at(0, 22.5) ;
			
			\coordinate[label=below:$x$] (x) at (2,-0.4);			
			\coordinate[label=below:$a$] (a1) at (4,-0.4);
			\coordinate[label=below:$y$] (y1) at (7.5,-0.4);
			\coordinate[label=below:$b$] (b1) at (11,0);
			\coordinate[label=below:$\tilde{y}$] (ytilde) at (14.5,0);
			\coordinate[label=below:$a$] (a2) at (18,-0.4);
			\coordinate[label=below:$y$] (y2) at (21.5,-0.4);
			\coordinate[label=below:$b$] (b2) at (25,0);
			\coordinate[label=below:$\tilde{y}$] (ytilde2) at (28.5,0);
			\coordinate[label=below:$b$] (b2) at (32,0);
			\coordinate[label=below:$z$] (z) at (34,-0.4);			
			
			\coordinate[label=left:{\rotatebox{90}{$x$}}] (xY) at (0,2);
			\coordinate[label=left:{\rotatebox{90}{$a$}}] (a1Y) at (0,4);
			\coordinate[label=left:{\rotatebox{90}{$y$}}] (y1Y) at (0.2,7.5);
			\coordinate[label=left:{\rotatebox{90}{$b$}}] (bY) at (0.2,11);
			\coordinate[label=left:{\rotatebox{90}{$\tilde{y}$}}] (ytildeY) at (0.2,14.5);
			\coordinate[label=left:{\rotatebox{90}{$a$}}] (a2Y) at (0,18);
			\coordinate[label=left:{\rotatebox{90}{$z$}}] (zY) at (0,20);

			
			\coordinate (Of) at (0,0) {};
			\node (Uf) at (4,4) [element] {};
			\node (Vf) at (11,11) [element] {};
			\node (Wf) at (18,4) [element] {};
			\node (Xf) at (25,11) [element] {};
			\node (Zf) at (32,18) [element] {};
			\coordinate (ZZf) at (36,22) {};
			
			\node (Wg) at (18,18) [element] {};
			
			\draw ($(Wf) + (-1.5,3)$) node {$f$};
			\draw (Of) -- (Uf) -- (Vf) -- (Wf) -- (Xf) -- (Zf) -- (ZZf);
			
			\draw ($(Wg) - (2,1)$) node {$g$};
			\draw[dashed] (Vf) -- (Wg) -- (Xf);
			
			\draw[dotted] (Uf) --(4,0);
			\draw[dotted] (Vf) --(11,0);
			\draw[dotted] (Wg) --(18,0);
			\draw[dotted] (Xf) --(25,0);
			\draw[dotted] (Zf) --(32,0);
			
			\draw[dotted] (Zf) -- (Wg) -- (0,18);
			\draw[dotted] (Xf) -- (0,11);
			\draw[dotted] (Wf) -- (0,4);
		\end{tikzpicture}
	\end{center}
	\caption{Distinct walks $f$ (solid) and $g$ (dashed and solid) on $u=xayb\ti{y}az$ such that $u^f = u^g$.}
	\label{fig:notUnique}
\end{figure}

For the converse, suppose for contradiction  that $u$ is a word of length $n$ containing no non-trivial palindromes, for which
there exist walks $f$ and $g$ such that 
$u^f = u^g$ but $f \neq g$. Let $u$, $f$ and $g$ be chosen so that $m = |u^f| = |u^g|$ is minimal.
If $f(i) = f(i+1)$ for some $i$, we have 
$g(i) = g(i+1)$, since otherwise, $u$ contains a repeated letter, and therefore a palindrome of length 2, contrary to assumption. 
But if both $f(i) = f(i+1)$ and $g(i) = g(i+1)$, then
the functions $f'= f/[i,i+1]$ and $g'= g/[i,i+1]$ are defined, and are obviously walks, and moreover we have
$u^{f'} = u^{g'}$ and $f' \neq g'$, contradicting the minimality of $m$. Hence, we may assume that neither $f$ nor $g$ is ever stationary. We claim that 
$f$ and $g$ have the same waypoints. For if $i$ is an internal waypoint for $f$ but not for $g$, we have $f(i{-}1) = f(i{+}1)$, $u[g(i{-}1)] = u[f(i{-}1)]$ and
$u[g(i{+}1)] = u[f(i{+}1)]$, whence $u[g(i{-}1)] = u[g(i{+}1)]$, so that $u$ contains an odd, non-trivial palindrome centred at
$g(i)$, contrary to assumption. This establishes the claim that $f$ and $g$ have the same waypoints.
Since $u$ is certainly not itself a non-trivial palindrome and $f \neq g$,
the walks $f$ and $g$ must have at least one internal waypoint between them.
Now take a shortest leg of $f$ (which must also be a shortest leg of $g$), say $[j,j+\ell]$. Suppose first that $[j,j+\ell]$ is an internal leg (i.e.~$j <1$ and $j+\ell <m$). To visualize the situation suppose $V = j$ and $W= j+\ell$ in Fig.~\ref{fig:minimalLeg}.
Taking into account the legs
either side, we see that $f(j) = f(j+2\ell)$ and $g(j) = g(j+2\ell)$, and moreover that $f' = f/[j,j+2\ell]$ and
$g'= g/[j,j+2\ell]$ map $[1,m-2\ell]$ surjectively onto $[1,n]$. Clearly, $u^{f'} = u^{g'}$. But $f$ and $g$ have
the same waypoints over the interval $[j, j+2\ell]$, whence $f \neq g$ implies 
$f' \neq g'$, 
contradicting the minimality of $m$.
The cases where the shortest leg is terminal are handled similarly.
\end{proof}

\section{Words yielding the same results on distinct walks}
\label{sec:equations}	
In this section, we prove Theorem~\ref{theo:defects}, allowing us to characterize those words which are solutions of a given equation $u^f = u^g$, for
walks $f$ and $g$.

Let $f' \colon [1,m] \rightarrow [1,n]$ be a walk. If $1 \leq j \leq m$, then the function
$f\colon [1,m+1] \rightarrow [1,n]$ given by
\begin{equation*}
	f(i) = 
	\begin{cases}
		f'(i) & \text{if $i \leq j$}\\
		f'(i{-}1)  & \text{otherwise}
	\end{cases}
\end{equation*}
is also a walk, longer by one step. We call $f$ the {\em hesitation on} $f'$ {\em at} $j$, as it arises by executing $f'$ up to and including the $j$th step, then 
pausing for one step, before continuing as normal.
We next proceed to define an operation of {\em vacillation on} $f'$, also producing 
a strictly longer walk. This operation
has three forms, depending on whether it occurs at the start, in the middle, or at the end of the walk.
For any $k$ ($1 \leq k < m$), we define the \textit{initial vacillation on} $f'$ \textit{over} $[1,k{+}1]$ to be the walk
$f\colon [1,m{+}k] \rightarrow [1,n]$ given by
\begin{equation*}
	f(i) = 
	\begin{cases}
		f'(k{+}1{-}(i{-}1)) & \text{if $i \leq k+1$}\\
		f'(i{-}k)     & \text{otherwise.}
	\end{cases}
\end{equation*}
Thus $f$ arises by executing the first $k+1$ steps of $f'$ in reverse order and then continuing to execute $f'$ from the second step as normal.
Likewise, we define the \textit{final vacillation on} $f'$  \textit{over} $[m{-}k,m]$ to be the walk
$f\colon [1,m{+}k] \rightarrow [1,n]$ given by
\begin{equation*}
	f(i) = 
	\begin{cases}
		f'(i) & \text{if $i \leq m$}\\
		f'(m{-}(i{-}m))     & \text{otherwise.}
	\end{cases}
\end{equation*}
Thus $f$ arises by executing $f'$ as normal and then repeating the $k$ steps preceding the last in reverse order.
Finally, for any  $j$ ($1 < j < m$), and any $k$ 
($1 \leq k < j$), we define the \textit{internal vacillation on} $f'$ \textit{over} $[j{-}k, j]$ to be the walk
$f\colon [1,m{+}2k] \rightarrow [1,n]$ given by
\begin{equation*}
	f(i) = 
	\begin{cases}
		f'(i)         & \text{if $i \leq j$}\\
		f'(j{-}(i{-}j))    & \text{if $j < i \leq j+k$}\\
		f'(i - 2k) & \text{otherwise.}
	\end{cases}
\end{equation*}
Thus $f$ arises by executing $f'$ up to the $j$th step, reversing the previous $k$ steps back to the $(j{-}k)$th step and then continuing
from the $(j-k+1)$th step as normal.
A \textit{vacillation} on $f'$ is an initial, internal or final vacillation on $f'$. 

Let $f' \colon [1,m] \rightarrow [1,n]$ again be a walk. We proceed to define an operation of {\em reflection} on $f'$, producing a stroll (not necessarily surjective) of the same length. 
For any $k$ ($1 \leq k < m$), we take the \textit{initial reflection on} $f'$ {\em over} $[1,k{+}1]$ to be the function $f$ defined on the domain $[1,m]$ by
setting 
\begin{equation*}
	f(i) =  
	\begin{cases}
		f'(k{+}1){-}(f'(i){-}f'(k{+}1)) & \text{if $i \leq k+1$}\\
		f'(i)                            & \text{otherwise.}
	\end{cases}
\end{equation*}
Thus $f$ arises by reflecting the segment of $f'$ over the interval $[1,k{+}1]$ in the horizontal axis positioned at height
$f'(k+1)$, and then continuing as normal
(Fig.~\ref{fig:initStationary}).  
Likewise, we take the \textit{final reflection on} $f'$ \textit{over} $[m{-}k,m]$ to be the function $f$ defined on $[1,m]$ by setting
\begin{equation*}
	f(i) = 
	\begin{cases}
		f'(m{-}k){-}(f'(i){-}f'(m{-}k)) & \text{if $i \geq m{-}k$}\\
		f'(i)                           & \text{otherwise.}
	\end{cases}
\end{equation*}
Thus $f$ arises by executing $f'$ as normal up to the 
$(m{-}k)$th step, and then thereafter reflecting the remaining segment of $f'$ in the horizontal axis positioned at height $f'(m{-}k)$
(Fig.~\ref{fig:finStationary}). 
Finally, for integers $j$, $k$ ($1 < j < m$, $1 \leq k \leq \min(j{-}1, m{-}j)$) such that $f'(j{-}k) = f'(j{+}k)$, 
the \textit{internal reflection on} $f'$ over $[j{-}k,j{+}k]$ is the function $f$ defined on $[1,m]$ by setting
\begin{equation*}
	f(i) = 
	\begin{cases}
		f'(j{-}k){-}(f'(i){-}f'(j{-}k)) & \text{if $j {-}k \leq i \leq j{+}k$}\\
		f'(i)                           & \text{otherwise.}
	\end{cases}
\end{equation*}
Thus $f$ arises by executing $f'$ up to the point $j{-}k$, then reflecting
the segment of $f'$ over the interval $[j{-}k,j+k]$ in the horizontal axis 
positioned at height $f'(j{-}k) = f'(j{+}k)$, thereafter executing $f'$ as normal (Fig.~\ref{fig:medStationary}). 
A \textit{reflection} on $f'$ is an initial, internal or final reflection on $f'$. 
As defined above, reflections can take values in the range $[-n{+}1, 2n{-}1]$;
accordingly, we call a reflection {\em proper} if all its values are within
the interval $[1,n]$, and in that case we take the resulting function to have co-domain $[1,n]$. We shall only ever be concerned with proper reflections in the sequel; and a proper
reflection on a walk (more generally, on a stroll) is evidently a stroll; there is no \textit{a priori} requirement for it to be surjective. 

Reflections are of most interest in connection with walks on words containing odd palindromes. 
Let $f'\colon [1,m] \rightarrow [1,n]$ be a stroll, and $u$ be a word of length $n$. We say that a reflection $f$ on $f'$ is {\em admissible for $u$} if it is either:
(i) an initial reflection over $[1,k{+}1]$, and $u$ has a palindrome of length $2k{+}1$ centred at $f'(k{+}1)$;
(ii) a final reflection over $[m{-}k, m]$, and $u$ has a palindrome of length $2k{+}1$ centred at $f'(m{-}k)$; or
(iii) an internal reflection over $[j{-}k,j{+}k]$, and $u$ has a palindrome of length $2k{+}1$ centred at $f'(j{-}k) = f'(j{+}k)$. 
We see by inspection of Fig.~\ref{fig:stationaryEdits} that, if $f$ is a reflection on $f'$ admissible for $u$, then $u^{f}= u^{f'}$.

Suppose now $f'$ and $g'$ are walks with domain $[1,m]$ and co-domain $[1,n]$. If $f$ and $g$ are hesitations on $f'$ and $g'$, respectively, at 
some common point, we say that the pair of walks $\langle f,g \rangle$ is a {\em hesitation} on the pair $\langle f',g' \rangle$; similarly, if $f$ and $g$ are 
vacillations on $f'$ and $g'$ over some common interval, we say that the pair of walks $\langle f,g \rangle$ is a {\em vacillation} on the pair $\langle f',g' \rangle$. Evidently, if $u$ is a word such that $u^{f'} = u^{g'}$ and 
$\langle f,g \rangle$ is a hesitation or vacillation on $\langle f',g' \rangle$, then $u^f = u^g$.
If now $f$ is a reflection on $f'$ over some interval, 
we say that $\langle f,g' \rangle$ is a {\em reflection on} $\langle f',g' \rangle$, 
and also that $\langle g',f \rangle$ is a {\em reflection on} $\langle g',f' \rangle$. Evidently, if the reflection in question is (proper and)  admissible for some word $u$
such that  $u^{f'} = u^{g'}$, then $u^f = u^{g'}$. 
Note that hesitations/vacillations on pairs of strolls are hesitations/vacillations  on \textit{both} of the strolls in question, while
reflections on pairs of strolls are reflections on \textit{either} of the strolls in question.

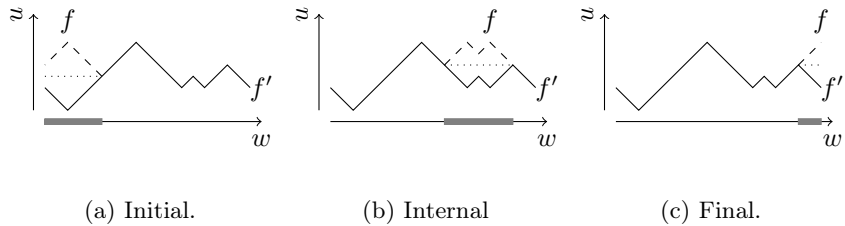
\begin{figure}
	\begin{subfigure}[t]{0.30\textwidth}
		\begin{center}
			\begin{tikzpicture}[scale= 0.15]
				\draw[->] (0,0) --  (19,0);
				\coordinate[label={$w$}] (wLabel) at (19,-3);
				\draw[->-=1] (-1,1) to (-1, 9.5);
				\coordinate[label={\rotatebox{90}{$u$}}] (wLabel) at(-2.5, 8) ;		
				\draw (0,3) -- (2,1) -- (8,7) -- (12,3) -- (13,4) -- (14,3) --(16,5) -- (18,3);
				\coordinate[label={$f'$}] (fLabel) at (19,1);
				\draw[dashed] (5,4) -- (2,7) -- (0,5);
				\coordinate[label={$f$}] (fLabel) at (2,7);
				\draw[dotted] (0,4) -- (5,4);
				\filldraw[draw=gray,fill=gray,opacity=0.3] (0, -0.25) rectangle (5,0.25);
			\end{tikzpicture}
		\end{center}
		\caption{Initial.}
		\label{fig:initStationary}
	\end{subfigure}
	\begin{subfigure}[t]{0.30\textwidth}
		\begin{center}
			\begin{tikzpicture}[scale= 0.15]
				\draw[->] (0,0) --  (19,0);
				\coordinate[label={$w$}] (wLabel) at (19,-3);
				\draw[->-=1] (-1,1) to (-1, 9.5);
				\coordinate[label={\rotatebox{90}{$u$}}] (wLabel) at(-2.5, 8) ;		
				\draw (0,3) -- (2,1) -- (8,7) -- (12,3) -- (13,4) -- (14,3) --(16,5) -- (18,3);
				\coordinate[label={$f'$}] (fLabel) at (19,1);
				\draw[dashed] (10,5) -- (12,7) -- (13,6) --(14,7) -- (16,5);
				\coordinate[label={$f$}] (fLabel) at (14,7);
				\draw[dotted] (10,5) -- (16,5);
				\filldraw[draw=gray,fill=gray,opacity=0.3] (10, -0.25) rectangle (16,0.25);
			\end{tikzpicture}
		\end{center}
		\caption{Internal}
		\label{fig:medStationary}
	\end{subfigure}
	\begin{subfigure}[t]{0.30\textwidth}
		\begin{center}
			\begin{tikzpicture}[scale= 0.15]
				\draw[->] (0,0) --  (19,0);
				\coordinate[label={$w$}] (wLabel) at (19,-3);
				\draw[->-=1] (-1,1) to (-1, 9.5);
				\coordinate[label={\rotatebox{90}{$u$}}] (wLabel) at(-2.5, 8) ;		
				\draw (0,3) -- (2,1) -- (8,7) -- (12,3) -- (13,4) -- (14,3) --(16,5) -- (18,3);
				\coordinate[label={$f'$}] (fLabel) at (19,1);
				\draw[dashed] (16,5) -- (18,7);
				\coordinate[label={$f$}] (fLabel) at (18,7);
				\draw[dotted] (16,5) -- (18,5);
				\filldraw[draw=gray,fill=gray,opacity=0.3] (16, -0.25) rectangle (18,0.25);
			\end{tikzpicture}
		\end{center}
		\caption{Final.}
		\label{fig:finStationary}
	\end{subfigure}
	\caption{The stroll $f$ (dashed and solid) is a reflection on the walk $f'$ (solid) over $I$ (shaded).}
	\label{fig:stationaryEdits}
\end{figure}

Now let $f'$ and $g'$ be walks, and suppose $u$ is a word of length $m$ such that $u^{f'} = u^{g'}$. We have seen that, if $\langle f,g \rangle$ is a hesitation or vacillation on
$\langle f',g'\rangle$, or is a reflection on
$\langle f',g'\rangle$ admissible for $u$, then $u^f = u^g$. The principal result of this section states that, for primitive words, this is essentially the only way in which we can arrive at distinct walks $f$ and $g$ such that $u^f = u^g$.
\begin{lemma}
	Let $u$ be a primitive word of length $n$, and let $f$ and $g$ be walks with domain $[1,m]$ and co-domain $[1,n]$ such that $u^f = u^g$.
	Then there exist sequences of walks $\set{f_s}_{s=0}^t$ and
	$\set{g_s}_{s=0}^t$, all having co-domain $[1,n]$, satisfying: 
	\textup{(i)} 
	$f_0 = g_0$ is monotone; 
	\textup{(ii)} 
	for all $s$ \textup{(}$1 \leq s < t$\textup{)},
	$
	\langle f_{s+1},g_{s+1} \rangle$ is a hesitation on $\langle f_s,g_s \rangle$, a vacillation on $\langle f_s,g_s \rangle$, or a reflection on $\langle f_s,g_s \rangle$ admissible for $u$;
	and 
	\textup{(iii)}
	$f_t= f$ and $g_t = g$.
	\label{lma:series}
\end{lemma}
\begin{proof}
	If $n =1$, there is nothing to prove, so assume otherwise.
	We observe first that since $u$ is primitive, and therefore contains no immediately  repeated letters,  $f(i)= f(i+1)$ if and only if $g(i)= g(i+1)$
	for all $i$ ($1 \leq i < m$). Hence, $\langle f, g \rangle$ arises by a sequence of hesitations on some pair of walks $\langle f'', g'' \rangle$ in which
	all legs are either ascending or descending.
	By replacing $f$ and $g$ by $f''$ and $g''$, respectively, we may thus assume that all legs of $f$ and $g$  are either ascending or descending. 
	
	We proceed by induction on the total number of internal waypoints in $f$ and $g$ taken together. If there are no internal waypoints on either $f$ or $g$, but 
	$f \neq g$, then $u$ is a non-trivial palindrome, contrary to the assumption
	that it is primitive. Hence, $f = g$, and setting $t = 0$ and $f_0 = g_0 = f = g$ establishes the base case.
	
	For the inductive case, we assume then that $f$ and $g$ have at least one internal waypoint between them, and therefore, since $u$ is primitive, at least one internal waypoint each. Let $[V,W]$ be a shortest leg of either $f$ or $g$, and assume (by exchanging $f$ and $g$ if necessary) that this is a leg of $f$. We consider first the case where 
	$[V,W]$ is internal. Let $\ell = W {-} V >0$, $U = 
	V{-}\ell$ and $X = W {+} \ell$. By the minimality of $\ell$, we have 
	$U \geq 1$, $X \leq m$, and $f$ is monotone over $[U,V]$, $[V,W]$ and $[W,X]$. We consider various possibilities for waypoints on $g$ over the interval  $J = [U+1,X-1]$. Since $u$ is primitive, $g$ must have at least one such waypoint.
	
	Suppose first that $g$ has a waypoints at both $V$ and $W$. By the minimality of $\ell$, $g$ too is monotone over $[U,V]$, $[V,W]$ and $[V,X]$
	(Fig.~\ref{fig:walksVW}). Letting $f' = f/[U,W]$ and $g' = g/[U,W]$, we see that $u^{f'} = u^{g'}$; moreover, $f'$ and $g'$ are by inspection surjective. 
	But $\langle f,g \rangle$ is a vacillation on $\langle f',g' \rangle$, $f'$ has fewer waypoints than $f$, and $g'$ fewer waypoints than $g$. The result then follows by inductive hypothesis.
	
	Suppose next that $g$ has a waypoint at $V$, but nowhere else
	in the interval $J$ (Fig.~\ref{fig:walksVonly}). In that case, $g(U) = g(W)$, and since $W$ is a waypoint on $f$, we see that
	$u$ has a palindrome of length $2\ell +1$ centred at $g(W)$.  Let $g''$ be a reflection
	on $g$ over $[U,V]$, (dashed lines). Then $u^{g''} = u^{g}$. Thus, $\langle f,g \rangle$ is a reflection on $\langle f,g'' \rangle$ admissible for $u$. Moreover, we are guaranteed that
	$g''$ is surjective, because otherwise $u^f = u^g= u^{g''}$  would have a shorter 
	generator than $u$, contradicting the assumption that $u$ is primitive. 
	But $g''$ has waypoints in $J$ at exactly $V$ and $W$, so that, 
	proceeding as in the previous case, $\langle f,g'' \rangle$ is a 
	vacillation on $\langle f',g' \rangle$, where $f' = f/[U,W]$ and $g= g/[U,W] = g''/[U,W]$. 
	Moreover, $u^{f'} = u^{g'}$, and $f'$ and $g'$ are by inspection surjective if $f$ and $g''$ are. But $f'$ and $g'$  
	have between them fewer waypoints than $f$ and $g$. The result then follows by inductive hypothesis. An exactly similar argument applies if $g$ has a waypoint at $W$, but nowhere else in the interval $J$. 
	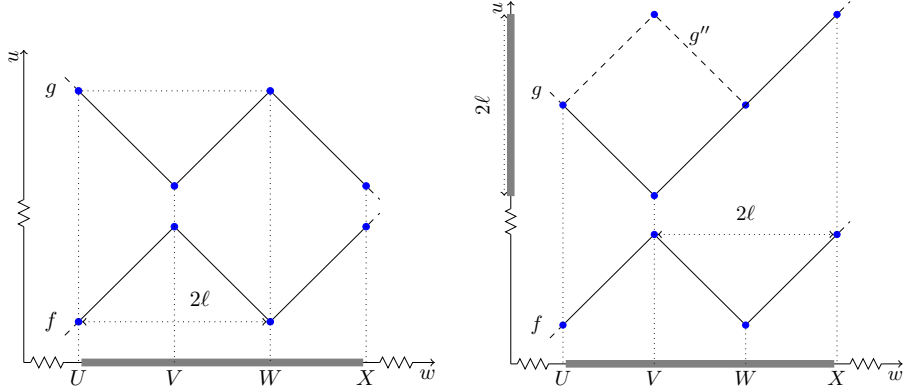
\begin{figure}
		\begin{subfigure}[t]{0.5\textwidth}
			\begin{center}


\resizebox{6cm}{!}{					
				\begin{tikzpicture}[scale= 0.22]
					
					\draw (3,0) to (26,0);
					\draw (0,0) -- (0.5,0);
					\draw[snake=zigzag,segment length = 5] (26,0) -- (28.5,0);
					\draw[snake=zigzag,segment length = 5] (0.5,0) -- (3,0);
					\draw[->-=1] (28.5,0) -- (30,0);
					\coordinate[label=below:{$w$}] (wLabel) at (29.5,0);
					
					\draw[-] (0,0) to (0, 10);
					\draw[snake=zigzag,segment length = 5] (0,10) -- (0,12);
					\draw[->-=1] (0,12) to (0, 23);
					\coordinate[label={\rotatebox{90}{$u$}}] (wLabel) at(-0.75, 21.5) ;
					
					\coordinate[label=below:$U$] (U) at (4,0);
					\coordinate[label=below:$V$] (V) at (11,0);
					\coordinate[label=below:$W$] (W) at (18,0);
					\coordinate[label=below:$X$] (X) at (25,0);
					
					\filldraw[draw=gray,fill=gray,opacity=0.3] (4.25, -0.25) rectangle (24.75,0.25);

					
					\node (Uf) at (4,3) [element] {};
					\node (Vf) at (11,10) [element] {};
					\node (Wf) at (18,3) [element] {};
					\node (Xf) at (25,10) [element] {};
					
					\node (Ug) at (4,20) [element] {};
					\node (Y) at (11,13) [element] {};
					\node (Z) at (18,20) [element] {};
					\node (Xg) at (25,13) [element] {};
					
					\draw ($(Uf) - (2,0)$) node {$f$};
					\draw[dashed] ($(Uf) - (1,1)$) -- (Uf);
					\draw (Uf) -- (Vf) -- (Wf) -- (Xf);
					\draw[dashed] (Xf) -- ($(Xf) + (1,1)$);
					
					\draw ($(Ug) - (2,0)$) node {$g$};
					\draw[dashed] ($(Ug) + (-1,1)$) -- (Ug);
					\draw (Ug) -- (Y) -- (Z) -- (Xg);
					\draw[dashed] (Xg) -- ($(Xg) + (1,-1)$);
					
					\draw[dotted] (Ug) --(U);
					\draw[dotted] (Y) --(V);
					\draw[dotted] (Z) --(W);
					\draw[dotted] (Xg) --(X);
					
					\draw[<->,dotted] (Uf) -- node[label=above right:{$2\ell$}] {} (Wf);
					\draw[dotted] (Ug) -- (Z);
				\end{tikzpicture}
}
			\end{center}
			\caption{The waypoints of $g$ on $J$ are $V$ and $W$.}
			\label{fig:walksVW}
		\end{subfigure}
		\begin{subfigure}[t]{0.5\textwidth}
			\begin{center}
\resizebox{6cm}{!}{	
				\begin{tikzpicture}[scale= 0.22]
					
					\draw (3,0) to (26,0);
					\draw (0,0) -- (0.5,0);
					\draw[snake=zigzag,segment length = 5] (26,0) -- (28.5,0);
					\draw[snake=zigzag,segment length = 5] (0.5,0) -- (3,0);
					\draw[->-=1] (28.5,0) -- (30,0);
					\coordinate[label=below:{$w$}] (wLabel) at (29.5,0);
					
					\draw[-] (0,0) to (0, 10);
					\draw[snake=zigzag,segment length = 5] (0,10) -- (0,12);
					\draw[->-=1] (0,12) to (0, 28);
					\coordinate[label={\rotatebox{90}{$u$}}] (wLabel) at(-1, 26.5) ;
					
					\coordinate[label=below:$U$] (U) at (4,0);
					\coordinate[label=below:$V$] (V) at (11,0);
					\coordinate[label=below:$W$] (W) at (18,0);
					\coordinate[label=below:$X$] (X) at (25,0);
					
					\filldraw[draw=gray,fill=gray,opacity=0.3] (4.25, -0.25) rectangle (24.75,0.25);
					\draw[<->,dotted] (-0.5,13) -- node [label=left:{\rotatebox{90}{$2\ell$}}] {} (-0.5,27);			
					
					\filldraw[draw=gray,fill=gray,opacity=0.3] (-0.25,13) rectangle (0.25,27);

					
					\node (Uf) at (4,3) [element] {};
					\node (Vf) at (11,10) [element] {};
					\node (Wf) at (18,3) [element] {};
					\node (Xf) at (25,10) [element] {};
					
					\node (Ug) at (4,20) [element] {};
					\node (Y) at (11,13) [element] {};
					\node (Ygpp) at (11,27) [element] {};
					\node (Z) at (18,20) [element] {};
					\node (Xg) at (25,27) [element] {};
					
					\draw ($(Uf) - (2,0)$) node {$f$};
					\draw[dashed] ($(Uf) - (1,1)$) -- (Uf);
					\draw (Uf) -- (Vf) -- (Wf) -- (Xf);
					\draw[dashed] (Xf) -- ($(Xf) + (1,1)$);
					
					\draw ($(Ug) - (2,-1)$) node {$g$};
					\draw[dashed] ($(Ug) + (-1,1)$) -- (Ug);
					\draw (Ug) -- (Y) -- (Xg);
					\draw[dashed] (Xg) -- ($(Xg) + (1,1)$);
					
					\draw[dashed] (Ug) -- (Ygpp) -- node[label=above:{$g''$}] {} (Z);
					
					\draw[dotted] (Ug) --(U);
					\draw[dotted] (Y) --(V);
					\draw[dotted] (Wf) --(W);
					\draw[dotted] (Xg) --(X);
					
					\draw[<->,dotted] (Vf) -- node[label=above:{$2\ell$}] {} (Xf);			
					
				\end{tikzpicture}
}
			\end{center}
			\caption{The only waypoint of $g$ on $J$ is $V$.}
			\label{fig:walksVonly}
		\end{subfigure}
		\caption{Proof of Lemma~\ref{lma:series}; shortest leg $[V,W]$ is internal, $g$ has a waypoint at $V$.}
		\label{walks}
	\end{figure}
	
	Henceforth, then, we may assume that $g$ has 
	a waypoint in $J$ other than $V$ or $W$.
	Let $Y$ be such a waypoint whose distance $k$ from either of $V$ or $W$ is minimal (thus $k >0$). By reversing the walks $f$ and $g$ if necessary, we may assume that $|Y - V| = k$. We claim that $k = \ell/2$. To see this, suppose first that $k \leq \ell/3$. Then we have the situation
	encountered in Fig.~\ref{fig:close} (setting $v = u$), contradicting the supposition that $u$ is primitive. Suppose next that $\ell/3 < k < \ell/2$.
	Then we have the situation encountered in Figs.~\ref{fig:middling} and~\ref{fig:threePossibilities}, again contradicting the supposition that $u$ is primitive. Finally, suppose $k > \ell/2$; necessarily, $Y < V$. By the minimality of $k$, $g$ has no waypoints over $[V,(W+X)/2]$, and
	we have the situation encountered in Fig.~\ref{fig:UYbig}, again contradicting the supposition that $u$ is primitive. This establishes $k = \ell/2$.
	
	Henceforth, then, we may assume that $g$ has 
	a waypoint at one of $(U+V)/2$ or $(V+W)/2$. By minimality of $\ell$ and $k$, $g$ does not have a waypoint at $V$. Further, if 
	$g$ has no waypoint at $(V+W)/2$, by the minimality of $\ell$, $(U+V)/2$ is the only waypoint of $g$ over $[X,W]$, and we have the situation of
	Fig.~\ref{fig:walksVplusWover2}, contradicting the supposition that $u$ is primitive.
	\begin{figure}
		\begin{subfigure}[t]{0.5\textwidth}
			\begin{center}
\resizebox{6cm}{!}{	
				\begin{tikzpicture}[scale= 0.22]
					\draw (3,0) to (26,0);
					\draw (0,0) -- (0.5,0);
					\draw[snake=zigzag,segment length = 5] (26,0) -- (28.5,0);
					\draw[snake=zigzag,segment length = 5] (0.5,0) -- (3,0);
					\draw[->-=1] (28.5,0) -- (30,0);
					\coordinate[label=below:{$w$}] (wLabel) at (29.5,0);
					
					\draw[-] (0,0) to (0, 10);
					\draw[snake=zigzag,segment length = 5] (0,10) -- (0,12);
					\draw[->-=1] (0,12) to (0, 25);
					\coordinate[label={\rotatebox{90}{$u$}}] (wLabel) at(-0.75, 23) ;
					
					\coordinate[label=below:$U$] (U) at (4,0);
					\coordinate (UplusVover2) at (7.3,0);
					\coordinate[label=below:$V$] (V) at (11,0);
					\coordinate (VplusWover2) at (14.3,0);
					\coordinate[label=below:$W$] (W) at (18,0);
					\coordinate[label=below:$X$] (X) at (25,0);
					

					
					\node (Uf) at (4,3) [element] {};
					\node (Vf) at (11,10) [element] {};
					\node (Wf) at (18,3) [element] {};
					\node (Xf) at (25,10) [element] {};
					
					\node (Ug) at (4,16.5) [element] {};
					\node (Y) at (7.5,13) [element] {};
					\node (Z) at (18,23.5) [element] {};
					\coordinate (VplusWover2g) at (14.3,20);
					
					\draw ($(Uf) - (2,0)$) node {$f$};
					\draw[dashed] ($(Uf) - (1,1)$) -- (Uf);
					\draw (Uf) -- (Vf) -- (Wf) -- (Xf);
					\draw[dashed] (Xf) -- ($(Xf) + (1,1)$);
					
					\draw ($(Ug) - (2,0)$) node {$g$};
					\draw[dashed] ($(Ug) + (-1,1)$) -- (Ug);
					\draw (Ug) -- (Y) -- (Z);
					\draw[dashed] (Z) -- ($(Z) + (1,1)$);
					
					\filldraw[draw=blue,fill=blue,opacity= 0.3] (4,16.5)  -- (4,0) -- (7.5,0)  -- (7.5,13) -- cycle;
					\filldraw[draw=green,fill=green,opacity= 0.3] (7.5,0)  -- (7.5,13) -- (11,16.5) -- (11,0) -- cycle;
					\filldraw[draw=blue,fill=blue,opacity= 0.3] (11,16.5) -- (11,0) -- (14.5,0) -- (14.5, 20) -- cycle;
					\filldraw[draw=green,fill=green,opacity= 0.3] (14.5,0) -- (14.5, 20) -- (18, 23.5) -- (18,0)-- cycle;

					\draw[dotted] (Ug) --(U);
					\draw[dotted] (Vf) --(V);
					\draw[dotted] (Z) --(W);
					\draw[dotted] (Xf) --(X);
					
				\end{tikzpicture}
}
			\end{center}
			\caption{The only waypoint of $g$ on $[U,W]$\\ is $\frac{(U+V)}{2}$.}
			\label{fig:walksVplusWover2}
		\end{subfigure}
		\begin{subfigure}[t]{0.5\textwidth}
			\begin{center}
\resizebox{6cm}{!}{	
				\begin{tikzpicture}[scale= 0.22]
					
					\draw (3,0) to (26,0);
					\draw (0,0) -- (0.5,0);
					\draw[snake=zigzag,segment length = 5] (26,0) -- (28.5,0);
					\draw[snake=zigzag,segment length = 5] (0.5,0) -- (3,0);
					\draw[->-=1] (28.5,0) -- (30,0);
					\coordinate[label=below:{$w$}] (wLabel) at (29.5,0);
					
					\draw[-] (0,0) to (0, 10);
					\draw[snake=zigzag,segment length = 5] (0,10) -- (0,12);
					\draw[->-=1] (0,12) to (0, 25);
					\coordinate[label={\rotatebox{90}{$u$}}] (uLabel) at(-0.75, 23) ;
					
					\coordinate[label=below:$U$] (U) at (4,0);
					\coordinate[label=below:$\frac{U+V}{2}$] (UplusVover2) at (7.5,0);
					\coordinate[label=below:$V$] (V) at (11,0);
					\coordinate[label=below:$\frac{V+W}{2}$] (VplusWover2) at (14.5,0);
					\coordinate[label=below:$W$] (W) at (18,0);
					\coordinate[label=below:$\frac{W+X}{2}$] (WplusXover2) at (21.5,0);
					

					
					\node (Uf) at (4,3) [element] {};
					\coordinate (UplusVover2f) at (7.5,6.5) [element] {};
					\node (Vf) at (11,10) [element] {};
					\coordinate (VplusWover2f) at (14.5,6.5) [element] {};
					\node (Wf) at (18,3) [element] {};
					\node (WplusXover2f) at (21.5,6.5) [element] {};
					
					\node (Ug) at (4,16.5) [element] {};
					\coordinate (Vg) at (11,16.5) {};
					\node (Y) at (7.5,13) [element] {};
					\node (Z) at (14.5,20) [element] {};
					\coordinate (Wg) at (18,16.5) [element] {};
					\coordinate (WplusXover2g) at (21.5,13)  {};
					
					\node (Vfpp) at (11,3) [element] {};
					
					\node (VplusWover2gpp) at (14.5,13) [element] {};

					\draw ($(Uf) - (2,0)$) node {$f$};
					\draw[dashed] ($(Uf) - (1,1)$) -- (Uf);
					\draw (Uf) -- (Vf) -- (Wf) -- (WplusXover2f);
					\draw[dashed] (WplusXover2f) -- ($(WplusXover2f) + (1,1)$);
					
					\draw[dashed] (UplusVover2f) -- (Vfpp) -- (VplusWover2f);
					
					\draw ($(Ug) - (2,0)$) node {$g$};
					\draw[dashed] ($(Ug) + (-1,1)$) -- (Ug);
					\draw (Ug) -- (Y) -- (Z) -- (WplusXover2g);
					\draw[dashed] (WplusXover2g) -- ($(WplusXover2g) + (1,1)$);
					
					\draw[dashed] (Vg) -- (VplusWover2gpp) -- (Wg);
					
					\filldraw[draw=gray,fill=gray,opacity=0.3] (-0.25,13) rectangle (0.25,20);
					\filldraw[draw=gray,fill=gray,opacity=0.3] (-0.25,3) rectangle (0.25,10);

					\draw[dotted] (Ug) --(U);
					\draw[dotted] (Y) --(UplusVover2);
					\draw[dotted] (Vg) --(V);
					\draw[dotted] (VplusWover2g) -- (VplusWover2);
					\draw[dotted] (Wg) --(W);
					\draw[dotted] (WplusXover2g) --(WplusXover2);
					
				\end{tikzpicture}
}
			\end{center}
			\caption{Waypoints of $g$ on $J$ are $\frac{(U+V)}{2}$, $\frac{(V+W)}{2}$, $\frac{(W+X)}{2}$.}
			\label{fig:threeHalfway}
		\end{subfigure}	
		\caption{Proof of Lemma~\ref{lma:series}; shortest leg $[V,W]$ is internal, $g$ has a waypoint at $\frac{U+V}{2}$.}
		\label{fig:series2}
	\end{figure}
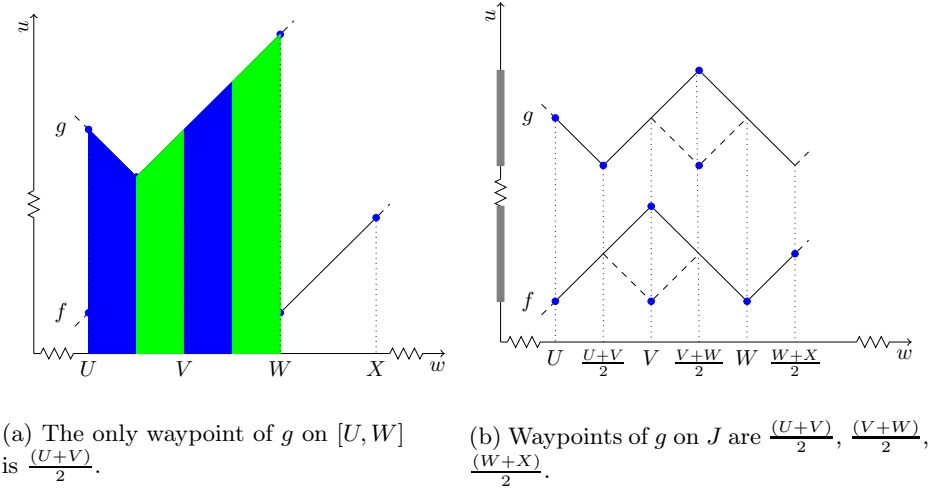	
	Hence $g$ has a waypoint at $(V+W)/2$, and hence no waypoint at $W$. By exactly similar reasoning, 
	$g$ therefore has waypoints at $(U+V)/2$ and $(W+X)/2$, 
	i.e.~the waypoints of $g$ on $J$ are exactly $(U+V)/2$, $(V+W)/2$ and $(W+X)/2$, 
	as illustrated in Fig.~\ref{fig:threeHalfway}. 
	
	Thus $u$ has palindromes of length $2k+1 = \ell + 1$ centred at $g(V) = g(W)$ and at $f((U+V)/2) = f((V+W)/2)$.
	So let $f''$ be a reflection on
	$f$ over $[(U{+}V)/2,(V{+}W)/2]$, and $g''$ a reflection on $g$ over $[V,W]$ (dashed lines). Thus,
	$f$ is a reflection of $f''$, admissible for $u$, and $g$ is a reflection of $g''$, also admissible for $u$.
	Moreover, we can again be sure that $f''$ and
	$g''$ are surjective, because no primitive generator of $u^f = u^g$ can be shorter than $u$.
	Now let $f''' = f''/[V,W]$ and $g''' = g''/[V,W]$, and let $f' = f'''/[U,V]$ and $g' = g'''/[U,V]$.
	Indeed,
	$\langle f'',g''\rangle$ is a vacillation on $\langle f''',g'''\rangle$, and 
	$\langle f''',g'''\rangle$ a vacillation on $\langle f',g'\rangle$. 
	Thus, $f' = f/[U,W]$ and $g' = g/[U,W]$, whence $u^{f'} = u^{g'}$; moreover, $f'$ and $g'$ are evidently surjective if
	$f''$ and $g''$ are. Finally,
	$f'$ and $g'$  
	have between them fewer waypoints than $f$ and $g$, and the result follows by inductive hypothesis. 
	
	It remains to deal with the cases where $[W,V]$ is either initial or final. We consider only the latter case: the former then follows by reversing the walks $f$ and $g$. Thus, we may assume $V = m - \ell$ and $W = m$. Let $U = V - \ell$. By minimality of $\ell$, $U \geq 1$, and, moreover, 
	$g$ is monotone on $[U,V]$ and $[V,W]$. If $g$ has a waypoint at $V$, then, again by the minimality of $\ell$, $g$ is monotone on $[U ,V]$
	and $[V,W]$. Now define
	$f' = f/[V,W]$ and $g' = g/[V,W]$. Evidently, $f'$ and $g'$ are surjective, $u^{f'} = u^{g'}$, and  $\langle f,g \rangle$ is a (final) vacillation on $\langle f',g'\rangle$.
	The result then follows by inductive hypothesis. 
	
	Hence, we may suppose that $g$ has no waypoint at $V$ and therefore no waypoint in $[V,W{-}1]$.
	Suppose $g$ also has no waypoint in $[U{+}1,V{-}1]$. Then $u$ has a palindrome of length $2\ell{+}1$ centred at $g(V)$. 
	Now let $g''$ be a (final) reflection on $g$ over $[V,W]$. Thus, $g$ is a reflection on $g''$, admissible for $u$, and $u^f = u^{g} = u^{g''}$.
	Again, $g''$ must be surjective, since $u$ is a primitive generator of $u^f = u^{g''}$.
	But $g''$ has a waypoint at $V$, and so we can proceed in the previous case, defining $f' = f/[V,W]$ and $g' = g''/[V,W]$, so that $u^{f'} = u^{g'}$, and $\langle f,g'' \rangle$ is a vacillation on $\langle f',g' \rangle$. Moreover, $f'$ and $g'$ between them have one fewer waypoint than $f$ and $g$, and so the result then follows by inductive hypothesis.

	Finally, we have the case where $g$ has a waypoint in $[U{+}1,V{-}1]$, and hence, by minimality of $\ell$, exactly one such waypoint, say at $Y$.
	Letting $k= Y {-} U$, it is easy to see that $k < \ell/2$, since otherwise, we would have the situation depicted in Fig.~\ref{fig:walksY closetoVend},
	contradicting the assumption that
	$u$ is primitive.  	
	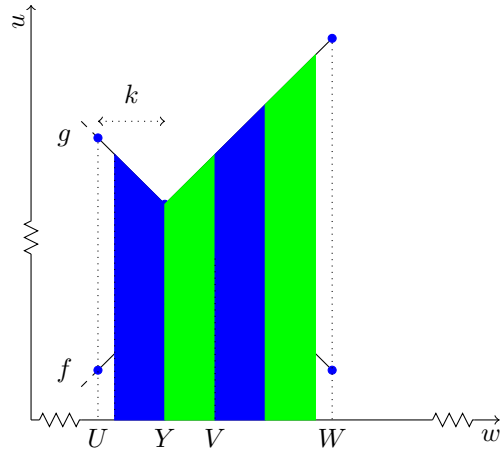
\begin{figure}
		\begin{center}
			\begin{tikzpicture}[scale= 0.22]
				\draw (3,0) to (24,0);
				\draw (0,0) -- (0.5,0);
				\draw[snake=zigzag,segment length = 5] (24,0) -- (26.5,0);
				\draw[snake=zigzag,segment length = 5] (0.5,0) -- (3,0);
				\draw[->-=1] (26.5,0) -- (28,0);
				\coordinate[label=below:{$w$}] (wLabel) at (27.5,0);
				
				\draw[-] (0,0) to (0, 10);
				\draw[snake=zigzag,segment length = 5] (0,10) -- (0,12);
				\draw[->-=1] (0,12) to (0, 25);
				\coordinate[label={\rotatebox{90}{$u$}}] (wLabel) at(-0.75, 23) ;
				
				\coordinate[label=below:$U$] (U) at (4,0);
				\coordinate (Uplus) at (5,0);
				\coordinate[label=below:$Y$] (Y) at (8,0);
				\coordinate[label=below:$V$] (V) at (11,0);
				\coordinate[label=below:$W$] (W) at (18,0);
				
				
				\node (Uf) at (4,3) [element] {};
				\node (Vf) at (11,10) [element] {};
				\node (Wf) at (18,3) [element] {};
				
				\node (Ug) at (4,17) [element] {};
				\coordinate (Uplusg) at (5,16);
				\node (Yg) at (8,13) [element] {};
				\node (Z) at (18,23) [element] {};
				
				\draw ($(Uf) - (2,0)$) node {$f$};
				\draw[dashed] ($(Uf) - (1,1)$) -- (Uf);
				\draw (Uf) -- (Vf) -- (Wf);
				
				\draw ($(Ug) - (2,0)$) node {$g$};
				\draw[dashed] ($(Ug) + (-1,1)$) -- (Ug);
				\draw (Ug) -- (Yg) -- (Z);
				
				\filldraw[draw=blue,fill=blue,opacity= 0.3] (5,16)  -- (5,0) -- (8,0)  -- (8,13) -- cycle;
				\filldraw[draw=green,fill=green,opacity= 0.3] (8,0)  -- (8,13) -- (11,16) -- (11,0) -- cycle;
				\filldraw[draw=blue,fill=blue,opacity= 0.3] (11,16) -- (11,0) -- (14,0) -- (14, 19) -- cycle;
				\filldraw[draw=green,fill=green,opacity= 0.3] (14,0) -- (14, 19) -- (17, 22) -- (17,0)-- cycle;
				
				\draw[<->,dotted] ($(Ug) + (0,1)$) -- node [label=above:{$k$}] {} ($(Ug) + (4,1)$);				
				
				\draw[dotted] (Ug) --(U);
				\draw[dotted] (Uplusg) --(Uplus);
				\draw[dotted] (Vf) --(V);
				\draw[dotted] (Z) --(W);
			\end{tikzpicture}
		\end{center}
		\caption{The only waypoint of $g$ on $[U,V]$ is $Y$, with $Y{-}U = k \geq \ell/2$.}
		\label{fig:walksY closetoVend}
	\end{figure}
	Thus we have the situation
	depicted in Fig.~\ref{fig:end}. Furthermore, letting $Z = W{-}k$,
	we see from the waypoints at $Y$ and $V$ that $u$ has a palindrome of length $2k{+}1$ centred at
	$g(Z)$, and also at $f(Z)$. (Note that $g(Z)$ and $f(Z)$ may or may not be the same points.) 
	\begin{figure}
		\begin{subfigure}[t]{0.5\textwidth}
			\begin{center}
\resizebox{6cm}{!}{	
				\begin{tikzpicture}[scale= 0.22]
					\draw (3,0) to (26,0);
					\draw (0,0) -- (0.5,0);
					\draw[snake=zigzag,segment length = 5] (26,0) -- (28.5,0);
					\draw[snake=zigzag,segment length = 5] (0.5,0) -- (3,0);
					\draw[->-=1] (28.5,0) -- (30,0);
					\coordinate[label=below:{$w$}] (wLabel) at (29.5,0);
					
					\draw[-] (0,0) to (0, 10);
					\draw[snake=zigzag,segment length = 5] (0,10) -- (0,12);
					\draw[->-=1] (0,12) to (0, 27);
					\coordinate[label={\rotatebox{90}{$u$}}] (wLabel) at(-0.75, 25) ;	
					\filldraw[draw=gray,fill=gray,opacity=0.3] (-0.25, 20.5) rectangle (0.25,24.5);
					\draw[<->,dotted] (-0.75,20.5) -- node[label=left:{\rotatebox{90}{$2k$}}] {} (-0.75,24.5);
					
					\filldraw[draw=gray,fill=gray,opacity=0.3] (-0.25, 3) rectangle (0.25,7);
					\draw[<->,dotted] (-0.75,3) -- node[label=left:{\rotatebox{90}{$2k$}}] {} (-0.75,7);

					\coordinate[label=below:$U$] (U) at (4,0);
					\coordinate[label=below:$Y$] (Y) at (6,0);				
					\coordinate[label=below:$V$] (V) at (11,0);
					\coordinate[label=below:$Z$] (Z) at (16,0);
					\coordinate[label=below:$W$] (W) at (18,0);

					\filldraw[draw=gray,fill=gray,opacity=0.3] (4.25, -0.25) rectangle (24.75,0.25);

					
					\node (Uf) at (4,3) [element] {};
					\node (Vf) at (11,10) [element] {};
					\node (Wf) at (18,3) [element] {};
					
					\node (Ug) at (4,14.5) [element] {};
					\node (Yg) at (6,12.5) [element] {};
					\coordinate (Vg) at (11,17.5) {};
					\node (Wg) at (18,24.5) [element] {};
					\coordinate (Zg) at (16,22.5);
					
					\draw ($(Uf) - (2,0)$) node {$f$};
					\draw[dashed] ($(Uf) - (1,1)$) -- (Uf);
					\draw (Uf) -- (Vf) -- (Wf);
					\draw[dashed] (Wf) -- ($(Wf) + (1,-1)$);
					
					\draw ($(Ug) - (2,0)$) node {$g$};
					\draw[dashed] ($(Ug) + (-1,1)$) -- (Ug);
					\draw (Ug) -- (Yg) -- (Wg);
					\draw[dashed] (Wg) -- ($(Wg) + (1,1)$);

					\filldraw[draw=blue,fill=blue,opacity= 0.3] (4,14.5)  -- (4,0) -- (6,0)  -- (6,12.5) -- cycle;
					\filldraw[draw=green,fill=green,opacity= 0.3] (6,0)  -- (6,12.5) -- (8,14.5) -- (8,0) -- cycle;
					\filldraw[draw=blue,fill=blue,opacity= 0.3] (16, 0) -- (16,22.5) -- (14,20.5) -- (14, 0) -- cycle;
					\filldraw[draw=green,fill=green,opacity= 0.3] (18,24.5) -- (18, 0) -- (16, 0) -- (16,22.5)-- cycle;

					\draw[dotted] (Ug) --(U);
					\draw[dotted] (Yg) --(Y);
					\draw[dotted] (Vg) --(V);
					\draw[dotted] (Wg) --(W);
					\draw[dotted] (Zg) --(Z);
					
					
					\draw[<->,dotted] ($(Ug) + (0,1)$) -- node [label=above:{$2k$}] {} ($(Ug) + (4,1)$);
					\draw[<->,dotted] (Uf) -- node [label=above right:{$\ell$}] {} ($(Uf) + (7,0)$);
					
					;			\end{tikzpicture}
}
			\end{center}
			\caption{The only waypoint of $g$ on $[U{+}1,W{-}1]$ is $Y$.}
			\label{fig:end}
		\end{subfigure}
		\begin{subfigure}[t]{0.5\textwidth}
			\begin{center}
\resizebox{6cm}{!}{	
				\begin{tikzpicture}[scale= 0.22]
					\draw (3,0) to (26,0);
					\draw (0,0) -- (0.5,0);
					\draw[snake=zigzag,segment length = 5] (26,0) -- (28.5,0);
					\draw[snake=zigzag,segment length = 5] (0.5,0) -- (3,0);
					\draw[->-=1] (28.5,0) -- (30,0);
					\coordinate[label=below:{$w$}] (wLabel) at (29.5,0);
					
					\draw[-] (0,0) to (0, 10);
					\draw[snake=zigzag,segment length = 5] (0,10) -- (0,12);
					\draw[->-=1] (0,12) to (0, 27);
					\coordinate[label={\rotatebox{90}{$u$}}] (wLabel) at(-0.75, 25) ;	
					\filldraw[draw=gray,fill=gray,opacity=0.3] (-0.25, 12.5) rectangle (0.25,22.5);
					\draw[<->,dotted] (-0.75,12.5) -- node[label=left:{\rotatebox{90}{$2(\ell -k)$}}] {} (-0.75,22.5);
					
					\coordinate[label=below:$U$] (U) at (4,0);
					\coordinate[label=below:$Y$] (Y) at (6,0);				
					\coordinate[label=below:$V$] (V) at (11,0);
					\coordinate[label=below:$Z$] (Z) at (16,0);
					
					\filldraw[draw=gray,fill=gray,opacity=0.3] (4.25, -0.25) rectangle (24.75,0.25);

					
					\node (Uf) at (4,3) [element] {};
					\coordinate (Yf) at (6,4) {};
					\node (Vf) at (11,10) [element] {};
					\node (Zf) at (16,5) [element] {};
					
					\node (Ug) at (4,14.5) [element] {};
					\node (Yg) at (6,12.5) [element] {};
					\coordinate (Vg) at (11,17.5) {};
					\node (Zg) at (16,22.5) [element] {};
					
					\draw ($(Uf) - (2,0)$) node {$f'''$};
					\draw[dashed] ($(Uf) - (1,1)$) -- (Uf);
					\draw (Uf) -- (Vf) -- (Zf);
					
					\draw ($(Ug) - (2,0)$) node {$g'''$};
					\draw[dashed] ($(Ug) + (-1,1)$) -- (Ug);
					\draw (Ug) -- (Yg) -- (Zg);
					
					\filldraw[draw=blue,fill=blue,opacity= 0.3] (6,0)  -- (6,12.5) -- (11,17.5) -- (11,0) -- cycle;
					\filldraw[draw=green,fill=green,opacity= 0.3] (11,17.5)  -- (11,0) -- (16,0) -- (16,22.5) -- cycle;
					
					\draw[dotted] (Ug) --(U);
					\draw[dotted] (Yg) --(Y);
					\draw[dotted] (Vg) --(V);
					\draw[dotted] (Zg) --(Z);
					
					\draw[<->,dotted] ($(Yf) + (0,1)$) -- node [label=below:{$\ell - k$}] {} ($(Yf) + (5,1)$);
					
					;			\end{tikzpicture}
}
			\end{center}
			\caption{The intermediate functions $f'''$ and $g'''$.}
			\label{fig:reallyEnd}
		\end{subfigure}
		\caption{Proof of Lemma~\ref{lma:series}; shortest leg $[V,W]$ is final.}
		\label{fig:seriesLast}
	\end{figure}
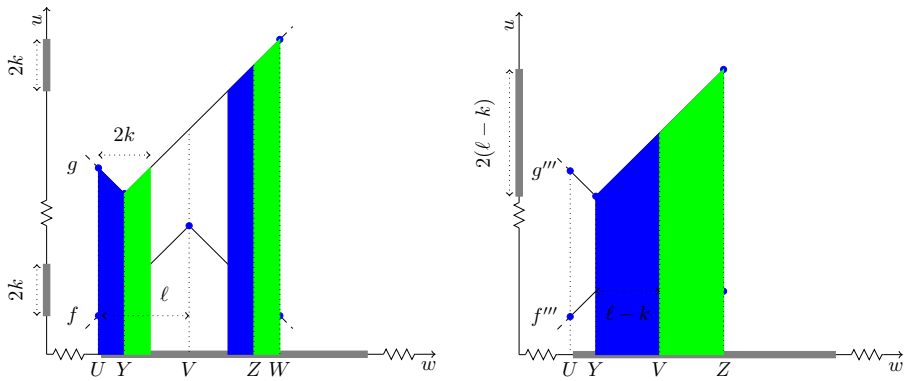
	Let $f''$  be a final reflection on $f$ over $[Z,W]$ and 
	$g''$ a final reflection on $g$ over the same interval. These reflections are thus admissible for $u$. Moreover,
	$f''$ and $g''$ are surjective, because $u$ is a primitive generator of $u^{f''} = u^{f} = u^g = u^{g''}$. 
	Now let $f''' = f''/[Z,W]$, and $g''' = g''/[Z,W]$. Thus,  
	$\langle f'',g'' \rangle$ is a final vacillation on
	$\langle f''',g''' \rangle$. 
	The functions $f'''$ and $g'''$ are depicted in
	Fig.~\ref{fig:reallyEnd}. But now $g'''$ has no waypoint on $[Y+1,Z-1]$, while $f'''$ has a waypoint $V$ in the middle of this interval.
	But this case has just been dealt with, and so the result again follows by inductive hypothesis.
\end{proof}

Lemma~\ref{lma:series} gives us everything we need for the proof of Theorem~\ref{theo:defects}.
Recall that, if $u$ is a word, then $\Delta_{{u}}$ is the defect set of $u$, and $\Delta^*_{{u}}$ is the equivalence closure of $\Delta_{{u}}$.

\newtheorem*{Restatetheo:defects}{Theorem~\ref{theo:defects}}
\begin{Restatetheo:defects}
Let ${u}$ be a primitive word of length $n$, and $f$, $g$ walks with domain $[1, m]$ and co-domain $[1, n]$. Then 
$u^f = u^g$ if and only if
$\langle f(i), g(i) \rangle \in \Delta^*(u)$ for all $i \in [1,m]$. 	
\end{Restatetheo:defects}
\begin{proof}
	The if-direction is almost trivial. Indeed, $\langle j, k\rangle \in 
\Delta_{{u}}$ certainly implies $u[j] = u[k]$, whence
$\langle j, k\rangle \in \Delta^*_{{u}}$ also implies $u[j] = u[k]$.
Thus, if
$\langle f(i), g(i) \rangle \in \Delta^*_{{u}}$
for all $i\in [1,m]$, then $u[f(i)] = u[g(i)]$  for all $i\in [1,m]$,
which is to say $u^f = u^g$.

For the only-if direction, we suppose that $u^f = u^g$.
By Lemma~\ref{lma:series}, we may decompose the pair of walks $\langle f, g \rangle$ into a series $\{ \langle f_s, g_s \rangle \}_{s=0}^t$
such that: (i) $f_0 = g_0$; (ii) for all $s$ ($0 \leq s < t$),
the pair $\langle f_{s+1}, g_{s+1} \rangle$ is obtained by performing a hesitation, vacillation, or an admissible (for $u$)
reflection on $ \langle f_s, g_s \rangle$; and
(iii) $\langle f_t, g_t \rangle = \langle f, g \rangle$. 
We establish that the following holds for all $s$ ($0 \leq s \leq t$):
\begin{equation}
	\langle f_s(i), g_s(i) \rangle \in \Delta^*_{{u}} \text{ for all $i$ in the domain of $f_s$ ($=$ the domain of $g_s$).} 
	\label{eq:DaumantasLemmaIH}
\end{equation}
Putting $s = t$ then secures the required condition.

We proceed by induction on $s$.
For the base case, where $s=0$, we have $f_0 = g_0$, and there is nothing to do. For the inductive step, we suppose~\eqref{eq:DaumantasLemmaIH}, and show that the same holds with $s$ replaced by $s+1$. We have three cases.

\medskip

\noindent
Case 1: $\langle f_{s{+}1}, g_{s{+}1} \rangle$ is obtained by a hesitation on $\langle f_s, g_s \rangle$ at $k$.
If $i \leq k$ then $f_{s{+}1}(i) = f_s(i)$ and $g_{s{+}1}(i) = g_{s}(i)$; and by~\eqref{eq:DaumantasLemmaIH}, $\langle f_s(i), g_s(i)\rangle \in \Delta^*_{{u}}$.
If $i > k$ then $f_{s{+}1}(i) = f_s(i{-}1)$ and $g_{s{+}1}(i) = g_s(i{-}1)$; and 
by~\eqref{eq:DaumantasLemmaIH}, $\langle f_s(i{-}1), g_s(i{-}1) \rangle \in \Delta^*_{{u}}$. In either case, $\langle f_{s{+}1}(i), g_{s{+}1}(i) \rangle \in \Delta^*_{{u}}$.

\medskip

\noindent
Case 2:  $\langle f_{s{+}1}, g_{s{+}1} \rangle$ is a vacillation on $\langle f_s, g_s \rangle$.
We consider the case of an internal vacillation over some interval over $[j {-} k, j]$; initial and final vacillations are handled similarly.  
Again, if $i \leq j$ then $f_{s{+}1}(i) = f_s(i)$ and $g_{s{+}1}(i) = g_s(i)$; and by~\eqref{eq:DaumantasLemmaIH},
$\langle f_s(i), g_s(i)\rangle \in \Delta^*_{{u}} $.
If $j < i \leq j {+} k$, then $f_{s{+}1}(i) = f_s(j{-}(i{-}j))$ and $g_{s{+}1}(i) = g_s(j{-}(i{-}j))$; and by~\eqref{eq:DaumantasLemmaIH},
$\langle f_s(j{-}(i{-}j)), g_s(j{-}(i{-}j))\rangle \in \Delta^*_{{u}}$.
Finally, if $i > j {+} k$, then $f_{s{+}1}(i) = f_s(i {-} 2k)$ and $g_{s{+}1}(i) = g_s(i {-} 2k)$; and by~\eqref{eq:DaumantasLemmaIH},
$\langle f_s(i {-} 2k), g_s(i {-} 2k)\rangle \in \Delta^*_{{u}}$.

\medskip

\noindent
Case 3: $\langle f_{s{+}1}, g_{s{+}1} \rangle$ is the result of a 
reflection on $f_s$ over some interval\linebreak $[j{-}k, j{+}k]$, with the reflection in question admissible for $u$.
By exchanging $f$ and $g$ if necessary, we may assume that $f_{s{+}1}$ is a reflection on $f_s$ over $[j{-}k, j{+}k]$, and $g_{s{+}1} = g_s$;
it does not matter for the ensuing argument whether the reflection in question is internal, initial or final.
If $i \not \in [j{-}k, j{+}k]$, then $f_{s{+}1}(i) = f_s(i)$ and $g_{s{+}1}(i) = g_{s}(i)$; and by~\eqref{eq:DaumantasLemmaIH},
$\langle f_s(i), g_s(i)\rangle\in \Delta^*_{{u}}$.
So suppose $i \in [j{-}k, k{+}j]$. 
Since the reflection over $[j{-}k, j{+}k]$ is admissible, the factor $u[f_{s}(j{-}k), f_s(j{+}k)]$ is a palindrome. Moreover, from
the definition of reflection, either $u[f_{s{+}1}(i), f_s(i)]$ or 
$u[f_{s}(i), f_{s{+}1}(i)]$  
is a palindromic factor of $u$ (depending on whether $f_{s{+}1}(i) \leq  f_{s}(i)$ or $f_{s{+}1}(i) \geq  f_{s}(i)$).
That is, either $\langle f_{s{+}1}(i), f_s(i) \rangle \in \Delta_{{u}}$ or $\langle f_{s}(i), f_{s{+}1}(i) \rangle \in \Delta_{{u}}$.
But by~\eqref{eq:DaumantasLemmaIH}, $\langle f_s(i), g_s(i)\rangle = \langle f_s(i), g_{s{+}1}(i)\rangle \in \Delta^*_{{u}}$. Hence
$\langle f_{s{+}1}(i), g_{s{+}1}(i)\rangle \in \Delta^*_{{u}}$, again as required. 
This concludes the induction, and hence the proof of the only-if direction.
\end{proof}
\begin{corollary}
	Let $v_1$ and $v_2$ be primitive words of length $n$. 
	Then $v_1$ and $v_2$ satisfy the same equations $u^f = u^g$, where
	$f$ and $g$ are walks with co-domain $[1,n]$, if and only if $\Delta(v_1) = \Delta(v_2)$.
	\label{cor:defects}
\end{corollary}
\begin{proof}
	The if-direction is immediate from Theorem~\ref{theo:defects}. For the only-if direction,
	suppose $v_1$ and $v_2$ satisfy the same equations $u^f = u^g$. If
	$v_1$ contains a palindrome of (necessarily odd) length, say, $2k+1$ centred at $i$, let $f$ and $g$ be walks as
	depicted  in Fig.~\ref{fig:notUnique}, diverging at $i$ and re-converging at $i+2k$. Thus $v_1^f = v_1^g$ and hence $v_2^f = v_2^g$. 
	But considering $f$ and $g$ over the interval $[i, i+k]$, the equation $v_2^f = v_2^g$ clearly implies that $v_2$ has a
	a  palindrome of length $2k+1$ centred at $f(i) = g(i) = i$, whence $\Delta(v_2) \supseteq \Delta(v_1)$. By symmetry, $\Delta(v_1) \supseteq \Delta(v_2)$
\end{proof}

\medskip

For a treatment of the problem of finding palindromes in words, see~\cite[Ch.~8]{words:cr02}.

\section{Primitive generators of some morphic words}
\label{sec:morphic}	
In this section, we prove Theorem~\ref{theo:kbonacci}, which states that, for each $k \geq 2$, all elements of the sequence $\set{\alpha^{(k)}_n}_{n \geq 1}$
from the $k$th onwards have the same primitive generator.  In the sequel, we employ decorated versions of $\alpha, \beta, \gamma$ as constants denoting words.

We work with an alternative, recursive definition of the words $\alpha^{(k)}_n$. For all $k \geq 1$, let $\beta_k = \beta'_k \cdot k$, where 
$\beta'_k$ is recursively defined by setting
$\beta'_1 = \varepsilon$ and $\beta'_{k{+}1} = \beta'_k \cdot k \cdot \beta'_k$ for all $k \geq 1$.
Now define, for any $k \geq 2$ the sequence $\set{\alpha^{(k)}_n}_{n\geq 1}$
by declaring, for all $n \geq 1$: $\alpha^{(k)}_n = \beta_n$ if $n \leq k$, and $\alpha^{(k)}_n = \alpha^{(k)}_{n{-}1} \alpha^{(k)}_{n{-}2} \cdots \alpha^{(k)}_{n{-}k}$ otherwise.
A simple induction shows that this definition of the $\alpha^{(k)}_n$ coincides with that given in the introduction via morphisms. We remark that
$|\beta_k| = 2^{k{-}1}$, for all $k\geq 1$. 

We now define the primitive generators promised by Theorem~\ref{theo:kbonacci}. For all $k \geq2$, let
$\gamma'_k = (k{-}1) \cdot \beta'_{k{-}1}$ and  $\gamma_k = \gamma'_k \cdot k$. We remark that $|\gamma_k| = 2^{k{-}2}{+}1$, for all $k\geq 2$. The following two claims are 
easily proved by induction.
\begin{claim}
	For all $k \geq 2$, $\beta'_k$ is a palindrome over $\set{1, \dots, k{-}1}$ containing exactly one occurrence of $k{-}1$ \text{(}in the middle\text{)}; thus $\beta_k$ contains
	exactly one occurrence of $k{-}1$ \text{(}at position $|\beta_{k}|/2$\text{)} and exactly one occurrence of $k$ \text{(}at the end\text{)}. For all $k \geq 3$, $\gamma_k$ contains exactly one occurrence of each of $k$ \text{(}at the end\text{)}, $k{-}1$ 
	\text{(}at the beginning\text{)} and $k{-}2$ \text{(}in the middle\text{)}.
	\label{claim:betaPrimeGammaPrimePrime}	
\end{claim}
\begin{claim}
	For all $k \geq 2$, any position in the word $\gamma_k$ is either occupied by the letter 1 or is next to a position occupied by the letter 1.
	\label{claim:occursIn}
\end{claim}
\begin{claim}
	For all $k \geq 2$, $\gamma_k$ is primitive. 
	\label{claim:gammaPrimitive}
\end{claim}
\begin{proof}
	By induction on $k$. Certainly, $\gamma_2 = 12$ is primitive. For $k \geq 2$, by Claim~\ref{claim:betaPrimeGammaPrimePrime},
	$\gamma_{k+1} = k \cdot \beta'_{k{-}1} \cdot (k{-}1) \cdot \beta'_{k{-}1} \cdot (k+1)$ contains exactly one occurrence of each of $k{+}1$, $k$ and $k{-}1$. 
	Considering the
	forms given by the four cases of Lemma~\ref{lma:main}, we see that $\gamma_{k+1}$ does not have a prefix or suffix which is a non-trivial palindrome, and
	that any occurrence of either of the patterns $aa$ or $axb\tilde{x}axb$ must be contained in one of the embedded occurrences of $\beta'_{k{-}1}$ and hence in $\gamma_k$. By inductive hypothesis, $\gamma_k$ is primitive, and therefore does not contain either of these patterns. But then
	$\gamma_{k+1}$ is primitive by Lemma~\ref{lma:main}.
\end{proof}

\begin{claim}
	Let $k \geq 2$. For all $h$ ($1 \leq h \leq |\gamma'_k|$) such that $\gamma'_k[h] = 1$, there exists a walk $f$ 
	such that: \textup{(i)} $\beta'_{k} = (\gamma'_k)^{f}$;
	\textup{(ii)} $f(1) = h$; and 
	\textup{(iii)} $f(|\beta'_k|) = |\gamma'_k|$.
	\label{claim:kaybonacciWalk}
\end{claim}
\begin{proof}
	We proceed by induction on $k$. 
	For $k=2$ and $k= 3$, the result is trivial, since $\beta'_2 = \gamma'_2 = 1$, $\beta'_3 = 121$ and $\gamma'_3=21$. 
	Now suppose the claim holds for the value $k \geq 3$.
	For convenience, we write $m = |\beta'_k|$ and $n = |\gamma'_k| = |\beta'_{k{-}1}|{+}1$ (so $m = 2n{-}1$.). 
	Remembering that $\beta'_{k{+}1} = \beta_k \cdot (k{+}1) \cdot \beta_k$, and
	$\gamma'_{k{+}1} = k \cdot \beta'_{k{-}1} \cdot (k-1) \cdot \beta'_{k{-}1} = k \cdot \beta'_{k{-}1} \cdot \gamma'_{k}$, we have
	$|\beta'_{k{+}1}| = 2m{+}1$ and $|\gamma'_{k{+}1}| = 2n$.
	To show that the claim also holds for the value $k{+}1$, pick any $h'$ satisfying
	($1 \leq h' \leq 2n$) such that $\gamma'_{k{+}1}[h'] = 1$. We show that there exists a walk $f'\colon [1,2m{+}1] \rightarrow [1,2n]$ such that 
	$\beta'_{k{+}1} = (\gamma'_{k{+}1})^{f'}$, $f'(1) = h'$, and 
	$f'(2m+1) = 2n$.
	\begin{figure}  
		\begin{center}				
			\resizebox{!}{7cm}{\begin{tikzpicture}[scale=0.75]
					\draw[-] (0,0) -- (0,5.5);
					\draw[|-|] (-2.5,0) -- node[left] {\rotatebox{90}{$\gamma'_{k{+}1}$}} (-2.5,5.5);
					\draw[|-|] (-0.5,0.5) -- node[left] {\rotatebox{90}{$\beta'_{k}$}} (-0.5,5.5);
					\node[element,label=left:{\rotatebox{90}{$k$}}] (startVericalm1)  at (-0.5,0) {};
					\draw[|-|] (-1.5,0.5) -- node[left] {\rotatebox{90}{$\beta'_{k{-}1}$}} (-1.5,2.5);
					\node[element,label=left:{\rotatebox{90}{$(k{-}1)$}}] (midVericalm2)  at (-1.5,3) {};
					\draw[|-|] (-1.5,3.5) -- node[left] {\rotatebox{90}{$\beta'_{k{-}1}$}} (-1.5,5.5);
					
					\draw[-] (0,0) -- (11,0);
					\draw[|-|] (0,-0.875) -- node[above] {$\beta'_{k{+}1}$}(11,-0.875);
					\draw[|-|] (0,-1.75) -- node[above] {$\beta'_k$}(5,-1.75);
					\node[element,label=below:{$k$}] (midHorizontalm2)  at (5.5,-1.75) {};
					\node[element,label=below:{$(k{-}1)$}] (quarterHorizontalm2)  at (2.5,-1.75) {};
					\node[element,label=below:{$(k{-}1)$}] (threeQuarterHorizontalm2)  at (8.5,-1.75) {};
					\draw[|-|] (6,-1.75) -- node[above] {$\beta'_k$} (11,-1.75);
					
					\draw[dotted] (2.5,0) to (2.5,5.5);
					\draw[dotted] (5,0) to (5,5.5);
					\draw[dotted] (5.5,0) to (5.5,5.5);
					\draw[dotted] (6,0) to (6,5.5);
					\draw[dotted] (11,0) to (11,5.5);
					\draw[dotted] (8.5,0) to (8.5,5.5);
					
					\draw[dotted] (0,0.5) to (11,0.5);
					\draw[dotted] (0,3) to (11,3);
					\draw[dotted] (0,5.5) to (11,5.5);
					\coordinate[label={$1$}] (J) at (0,5.5);
					\coordinate[label={$J$}] (J) at (2.5,5.5);
					\coordinate[label={$m$}] (m) at (5,5.5);
					\coordinate[label={$J{+}(m{+}1)$}] (J) at (8.5,5.5);
					\coordinate[label={$2m{+}1$}] (J) at (11,5.5);

					\draw[domain= 0:2.5,smooth] plot (\x, {4.4 -\x*sin(5.6*\x*57.29577951308232)/1.82}); 
					
					\coordinate[label={$g$}] (flabel) at (3.75,4.3);
					\coordinate[label={$g'$}] (glabel) at (3.75,2);
					\coordinate[label={$f'$}] (gpplabel) at (7.5,2);
					\coordinate[label={$f'$}] (gpplabel) at (10,4.5);

					\node[element] (vstart) at (5,0.5) {};
					\node[element] (vmid) at (5.5,0) {};
					\node[element] (vend) at (6,0.5) {};
					\draw (vstart) -- (vmid) -- (vend);
					
					\node[element] (gstart) at (0,4.45) {};
					\node[element] (gmid) at (2.5,3) {};
					\node[element] (gend) at (5,5.5) {};
					\node[element] (fPrimeThreeQuarters) at (8.5,3) {};
					\node[element] (fPrimeEnd) at (11,5.5) {};
					
					\draw (gmid) -- (gend);	
					\draw[dashed] (gmid) -- (vstart);						
					\draw[dotted] (vend) --	(fPrimeThreeQuarters);
					\draw[dotted] (fPrimeThreeQuarters)	-- (fPrimeEnd);					
				\end{tikzpicture}
			}
		\end{center}	
		\caption{Proof of Claim~\ref{claim:kaybonacciWalk}: $g$ (solid lines) is a shifted copy of a walk $f$ on $\gamma'_k$ yielding $\beta'_k$; $g'$
			(solid and dashed lines) is a final reflection on $g$ over $[J,m]$; $f'$ (solid, dashed and dotted lines) is a walk on $\gamma'_{k{+}1}=
			k\cdot \beta'_k \cdot (k{+}1)$ yielding
			$\beta'_{k{+}1} = \beta'_k \cdot k \cdot \beta'_{k}$. }
		\label{fig:kaybonacci}
	\end{figure}
	
	Assume for the time being that $h' > n {+}1$, that is to say, $h'$ is a position in $\gamma'_{k{+}1} = k \cdot \beta'_{k{-}1} \cdot (k-1) \cdot \beta'_{k{-}1}$ occupied by a 1 and lying in the {\em second} copy of $\beta'_{k{-}1}$. Then $h= h'- n$ is a position in $\gamma'_{k} = (k-1) \cdot \beta'_{k{-}1}$ occupied by a 1, so by inductive hypothesis, let $f\colon [1,m] \rightarrow [1,n]$ be a walk
	such that $\beta'_{k} = (\gamma'_k)^{f}$, $f(1) = h$, and $f(m) = n$.
	By Claim~\ref{claim:betaPrimeGammaPrimePrime}, $\beta'_k$ contains exactly one occurrence of $k{-}1$ (this will be exactly in the middle), and $\gamma'_k$
	likewise contains exactly one occurrence of $k{-}1$ (this will be at the very beginning). Thus, $f$ reaches the value $1$ at just one
	point, namely $J = (m+1)/2$, and is otherwise strictly greater. (In fact, it is obvious that $f$ must be a straight line from $J$ onwards.)
	We first define a stroll $g\colon [1,m] \rightarrow [1,2n]$ given by $g(i) = f(i){+} n$ (Fig.~\ref{fig:kaybonacci}, solid lines). Thus, $g(1) = h'$, and $g(m) = 2n$. Moreover, $g$ reaches the value $n{+}1$ at just one point, namely $J = (m{+}1)/2$, and is otherwise strictly greater, as illustrated.
	Now let $g'$ be the (final) reflection on $g$ over the interval $[J, m]$ (Fig.~\ref{fig:kaybonacci}, first solid, then dashed lines). Thus,
	$g'$ is a stroll on $\gamma'_{k{+}1}$ satisfying $g'(1) = h'$ and $g'(m) =2$. Moreover, since $\beta'_{k} = \beta'_{k{-1}} \cdot (k-1) \cdot \beta'_{k{-}1}$
	is a palindrome, we see by inspection that $(\gamma'_{k{+}1})^{g'} = (\gamma'_{k{+}1})^{g} = (\gamma'_{k})^{f} = \beta'_{k}$. We now construct the desired walk $f' \colon
	[1,2m+1] \rightarrow [1,2n]$. 
	For $i \in [1,m]$, we set $f'(i)= g'(i)$.
	Since $f'(m) = g'(m) =2$, we set $f'(m+1) =1$, and then proceed to define $f'$ over the positions to the right, corresponding to the second copy of
	$\beta'_k$ in the word $\beta'_{k{+}1} = \beta'_k\cdot k \cdot \beta'_k$. But this we can do by drawing a straight line, as shown in (Fig.~\ref{fig:kaybonacci}).
	By inspection, $f'$ has the required properties.
	
	Finally, we consider the case where $h' \leq n{+}1$. Since $\gamma'_{k{+}1}[h']=1$ we in fact have $2 \leq h' \leq n$.
	And since $\beta'_{k}$ is a palindrome, we may replace $h'$ with the value $(n{+}1) + ((n{+}1){-}h)$ (i.e.~reflect in the horizontal axis
	at height $n{+}1$) and construct $f'$ as before. To re-adjust so that $f'(1)$ has the correct value, perform an initial reflection 
	on $f'$ over the interval $[1,J]$.
\end{proof}
\begin{claim}
	Let $k \geq 2$. For all $h$ ($1 \leq h \leq |\gamma_k|$) such that $\gamma_k[h] = 1$, there exists a walk $f$ 
	such that: \textup{(i)} $\gamma^f_k=\beta_{k}$;
	\textup{(ii)} $f(1) = h$; and 
	\textup{(iii)} $f(|\beta_k|) = |\gamma_k|$.
	\label{claim:kaybonacciWalkWithEnd}
\end{claim}
\begin{proof}
	Take the walk guaranteed by Claim~\ref{claim:kaybonacciWalk}, and, noting that the final letters of $\beta_k$ and $\gamma_k$ are both $k$, 
	extend $f$ by setting $f(|\beta_k|) = |\gamma_k|$.
\end{proof}

\begin{claim}
	Let $k \geq 2$. For all $h$ $(1 \leq h \leq |\gamma_n|)$ such that $\gamma_k[h] = 1$, and for all $p$ $(1 \leq p < k)$ there exists a walk $g$ 
	such that: \textup{(i)} $\gamma_k^{f} = \beta_{k} \beta_{k{-}1} \cdots \beta_{k{-}p}$ and \textup{(ii)} $f(1) = h$.
	\label{claim:kaybonacciStrollEarlyTerms}
\end{claim}
\begin{proof}
	Since $\beta_2\beta_1 = 121$ and $\gamma_2 = 12$, the claim is immediate for $k=2$. Hence we may assume $k \geq 3$.  
	By Claim~\ref{claim:kaybonacciWalkWithEnd}, let $f_1$ be a walk on $\gamma_k$ yielding $\beta_k$ with $f_1(1) = h$ and $f_1(|\beta_k|) = |\gamma_k|$.
	Since $f$ is a walk, and $\beta_k$ contains only one occurrence of $k$, we have $f_1(|\beta_k|{-}1) = |\gamma_k|{-}1$. Set $g_1 = f_1$.
	By Claim~\ref{claim:betaPrimeGammaPrimePrime}, $\beta'_{k{-}2}$ is a palindrome, whence
	$\gamma_k = (k{-}1) \cdot \beta'_{k{-}1} \cdot k = (k{-}1) \cdot \big(\beta'_{k{-}2} \cdot (k{-}2) \cdot \beta'_{k{-}2}\big) \cdot k = 
	\tilde{\gamma}_{k-1} \cdot \beta'_{k{-}2} \cdot k$.
	Noting that the penultimate position of $\gamma_{k{-}1}$ is always occupied by the letter 1,
	by Claim~\ref{claim:kaybonacciWalkWithEnd}
	let $f_2$ be a walk on $\gamma_{k-1}$ yielding the word $\beta_{k{-}1}$, with $f_2(1) = |\gamma_{k-1}|{-}1$. By Claim~\ref{claim:betaPrimeGammaPrimePrime}, $f_2(i) =1$
	only when $i = |\beta_{k{-}1}|/2$, since that is the only position of $\beta_{k{-}1}$ occupied by the 
	letter $k{-}2$. It follows that
	the function $\tilde{f}_2$ defined by
	$\tilde{f}_2(i) = |\gamma_{k{-}1}|{-}(f_2(i){-}1)$ is a walk on $\tilde{\gamma}_{k{-}1}$ yielding the word  $\beta_{k{-}1}$ with $\tilde{f}_2(1) = 2$, and achieving its
	maximum value $f(i) = |\gamma_{k{-}1}|$ only at $i = |\beta_{k{-}1}|/2$.
	Now regarding $\tilde{f}_2$ as a stroll on $\gamma_k = \tilde{\gamma}_{k-1} \cdot \beta'_{k{-}2} \cdot k$, let $\tilde{f}'_2$
	be the initial reflection on $\tilde{f}_2$
	over the interval $[1, |\beta_{k{-}1}|/2]$. Since $\gamma_k = (k-1) \cdot \beta'_{k{-}1} \cdot k$ , with $\beta'_{k{-}1}$ a palindrome, it follows that the stroll $\tilde{f}'_2$ on $\gamma_k$ also yields the same result as $\tilde{f}_2$, namely $\beta_{k{-}1}$.
	Now let $g_2$ be the result of appending $\tilde{f}'_2$ to $g_1$, as shown in the shaded part of Fig.~\ref{fig:kaybonacciStrollEarlyTerms}. 
	(Most of the curves drawn schematically here will actually be straight lines, but no matter.)
	Formally, we define
	$g_2: [1, |\beta_k| + |\beta_{k{-}1}|] \rightarrow [1, |\gamma_k|]$ to be the function:
	\begin{equation*}
		g_2(i) =
		\begin{cases}
			g_1(i) & \text{if $1 \leq i \leq |\beta_k|$}\\
			\tilde{f}'_2(i - |\beta_k|) & \text{if $|\beta_k| < i \leq |\beta_k| + |\beta_{k{-}1}|$}.
		\end{cases}
	\end{equation*}
	Since $g_1(|\beta_k|) = |\gamma_k|$ and $\tilde{f}'_2(1) = |\gamma_k|{-}1$, we see that
	$g_2$ is indeed a walk on $\gamma_k$ as shown (i.e.~with no jumps), yielding $\beta_k \beta_{k{-}1}$. We remark that $g_2(|\beta_k| + |\beta_{k{-}1}|) =1$.
	Notice that we needed to invert $f_2$ to yield $\tilde{f}_2$, so as to make the latter's reflection $\tilde{f}'_2$
	join up to the end of $g_1$ properly. 
	
	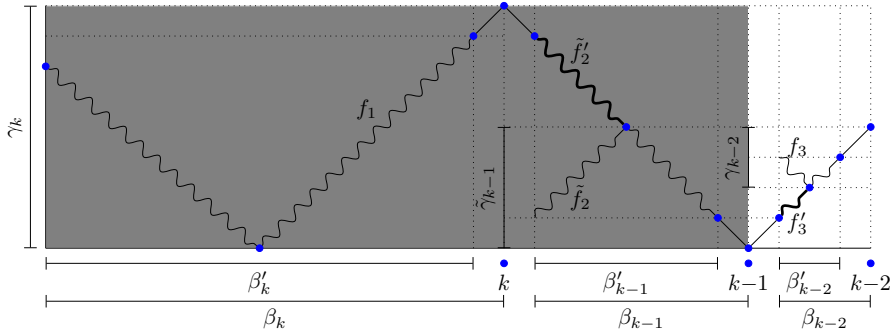
\begin{figure}  
		\begin{center}	
\resizebox{12cm}{!}{	
			\begin{tikzpicture}[scale=0.48]
				
				\filldraw[color=white,fill=gray,opacity=0.1] (1,1) rectangle (24,9);
				
				\draw (1,1) -- (28,1);
				\draw[dotted] (16,2) --(25,2);
				\draw[dotted] (24,3) --(28,3);
				\draw[dotted] (24,4) --(28,4);
				\draw[dotted] (16,5) --(28,5);
				\draw[dotted] (1,8) --(16,8);
				\draw[dotted] (1,9) --(28,9);
				
				\draw[|-|] (1,0.5) -- node[below] {$\beta'_k$}(15,0.5);
				\node[element,label=below:{$k$}] (k)  at (16,0.5) {};
				\draw[|-|] (17,0.5) -- node[below] {$\beta'_{k{-}1}$}(23,0.5);
				\node[element,label=below:{$k{-}1$}] (kmo)  at (24,0.5) {};
				\draw[|-|] (25,0.5) -- node[below] {$\beta'_{k{-}2}$}(27,0.5);
				\node[element,label=below:{$k{-}2$}] (kmt)  at (28,0.5) {};
				
				\draw[|-|] (1,-0.75) -- node[below] {$\beta_k$}(16,-0.75);
				\draw[|-|] (17,-0.75) -- node[below] {$\beta_{k{-}1}$}(24,-0.75);
				\draw[|-|] (25,-0.75) -- node[below] {$\beta_{k{-}2}$}(28,-0.75);
				

				\draw (1,1) -- (1,9); 
				\draw[dotted] (15,1) --(15,9);
				\draw[dotted] (16,1) --(16,9);
				\draw[dotted] (17,1) --(17,9);
				
				\draw[dotted] (23,1) --(23,9);
				\draw[dotted] (24,1) --(24,9);
				\draw[dotted] (25,1) --(25,9);
				
				\draw[dotted] (27,1) --(27,9);
				\draw[dotted] (28,1) --(28,9);
				
				
				\draw[|-|] (0.5,1) -- node[left] {\rotatebox{90}{$\gamma_k$}} (0.5,9);
				\draw[|-|] (16,1) -- node[left] {\rotatebox{90}{$\tilde{\gamma}_{k{-}1}$}} (16,5);
				\draw[|-|] (24,3) -- node[left] {\rotatebox{90}{$\gamma_{k{-}2}$}} (24,5);

				
				\node[element] (f1start) at (1,7){};
				\node[element] (f1mid) at (8,1){};
				\node[element] (f1pen) at (15,8){};
				\node[element] (f1end) at (16,9){};
				
				\node[element] (f2tstart) at (17,8){};
				\node[element] (f2tmid) at (20,5){};
				\node[element] (f2tpen) at (23,2){};
				\node[element] (f2tend) at (24,1){};
				
				\node[element] (f3start) at (25,2){};
				\node[element] (f3mid) at (26,3){};
				\node[element] (f3pen) at (27,4){};
				\node[element] (f3end) at (28,5){};

				\coordinate[label=$f_1$] (U) at (11.5,5);
				\draw [-,decorate,decoration=snake] (f1start) -- (f1mid);
				\draw [-,decorate,decoration=snake] (f1mid) -- (f1pen);
				\draw (f1pen) -- (f1end) -- (f2tstart); 
				
				\coordinate[label=$\tilde{f}_2$] (U) at (18.5,2);
				\coordinate[label=$\tilde{f}'_2$] (U) at (18.5,6.8);
				\draw [-,decorate,decoration=snake] (17,2) -- (f2tmid);
				\draw [-,decorate,decoration=snake,very thick] (f2tstart) -- (f2tmid);
				\draw [-,decorate,decoration=snake] (f2tmid) -- (f2tpen);
				\draw (f2tpen) -- (f2tend) -- (f3start); 
				
				\coordinate[label=$f_3$] (U) at (25.6,3.75);
				\coordinate[label=$f'_3$] (U) at (25.6,1.2);
				\draw [-,decorate,decoration=snake] (25,4) -- (f3mid);
				\draw [-,decorate,decoration=snake,very thick] (f3start) -- (f3mid);
				\draw [-,decorate,decoration=snake] (f3mid) -- (f3pen);
				\draw (f3pen) -- (f3end); 
				
			\end{tikzpicture}
}
		\end{center}
		\caption{Proof of Claim~\ref{claim:kaybonacciStrollEarlyTerms} (schematic drawing): thin lines depict $f_1$, $\tilde{f}_2$ and $f_3$; thick lines denote the 
			results $\tilde{f}'_2$ and $f'_3$ of performing initial reflections.}
		\label{fig:kaybonacciStrollEarlyTerms}
	\end{figure}
	
	We now repeat the above procedure, as
	shown in the unshaded part of Fig.~\ref{fig:kaybonacciStrollEarlyTerms}.
	By Claim~\ref{claim:kaybonacciWalkWithEnd}, 
	and noting that the penultimate position of $\gamma_{k{-}2}$ is occupied by the letter 1,
	let $f_3$ be a walk on $\gamma_{k-2}$ yielding the word $\beta_{k-2}$, with $f_3(1) = |\gamma_{k{-}2}|{-}1$. By Claim~\ref{claim:betaPrimeGammaPrimePrime}, 
	$f_3(i) =1$ only when $i= |\beta_{k{-}2}|/2$, since that is the only position of $\beta_{k{-}2}$ occupied by the letter $k-3$.  
	Observing 
	that $\gamma_{k{-}1}  = \tilde{\gamma}_{k-2} \cdot \beta'_{k{-}3} \cdot k$, and hence 
	$\tilde{\gamma}_{k{-}1}  =  k \cdot \tilde{\beta}'_{k{-}3} \cdot {\gamma}_{k-2}$, we
	see that, by shifting $f_3$ upwards by $|k \cdot \tilde{\beta}'_{k{-}2}|$, we can regard it as a stroll on $\gamma_k$. This (shifted) stroll reaches its minimum value
	$|k \cdot \tilde{\beta}'_{k{-}3}| +1$ exactly once in the middle of its range. 
	Let $f'_3$
	be the initial reflection on of this stroll
	over the interval $[1, |\beta_{k{-}2}|/2]$.
	Since $\tilde{\gamma}_{k{-}1} = (k-1) \cdot \beta'_{k{-}2} \cdot (k{-}2)$ with $\beta'_{k{-}2}$ a palindrome, we see by inspection that the stroll $f'_3$ 
	on $\gamma_{k}$ yields the same result as $f_3$, namely $\beta_{k{-}2}$.
	Now take $g_3$ to be the result of appending   
	${f}'_3$ to $g_2$, just as we earlier appended $\tilde{f}'_2$ to $g_1$. Thus, $g_3$ is a walk on $\gamma_k$ yielding $\beta_k \beta_{k{-}1} \beta_{k{-}2}$. Notice incidentally that
	$f_3$, unlike $f_2$, did not need to be inverted, so as to make its reflection $f'_3$ join up to the end of $g_2$. 
	
	Evidently, this process may be continued until we obtain the desired walk $g_{p{+}1}$ on $\gamma_k$ yielding $\beta_k \beta_{k{-}1} \cdots \beta_{k{-}p}$, with the inversion step (producing $\tilde{f}_h$ from $\tilde{f}_h$) required only when $h$ is even.
\end{proof}

We now prove Theorem~\ref{theo:kbonacci}, establishing by induction the following slightly stronger claim.
\begin{claim}
	Fix $k \geq 2$. For all $n \geq k$ and for all $h$ ($1 \leq h \leq |\gamma_k|$) such that $\gamma_n[h] = 1$,  there exists a walk $f$ 
	such that $\alpha^{(k)}_n = \gamma^f_k$ and $f(1) = h$.
	\label{claim:slightlyStronger}
\end{claim}
\begin{proof}
	If $n=k$, then $\alpha^{(k)}_n = \beta_k$, and the result is 
	immediate from Claim~\ref{claim:kaybonacciWalkWithEnd}.
	If $n=k{+}1$, then $\alpha^{(k)}_n = \beta_k \beta_{k{-}1} \cdots \beta_1$, and the result is 
	immediate from Claim~\ref{claim:kaybonacciStrollEarlyTerms}, setting $p=k{-}1$.
	
	For the inductive step we suppose $n \geq k{+}2$ and assume the result holds for values smaller than $n$. We consider first the slightly easier case where $n \geq 2k$.
	Set $h_1 = h$.
	Writing $\alpha^{(k)}_{n} = \alpha^{(k)}_{n{-}1} \cdots \alpha^{(k)}_{n{-}k}$, 
	by inductive hypothesis, let $g_1$ be a walk such that $\alpha^{(k)}_{n{-}1} = \gamma^{g_1}_k$ and $g_1(1) = h_1$.
	Now let $h'_1$ be the final value of $g_1$, that is, $g_1(|\alpha^{(k)}_{n{-}1}|) = h'_1$. By Claim~\ref{claim:occursIn}, there exists
	$h_2$ such that $|h_2 {-} h'_1| \leq 1$, and $\gamma_k[h_2] = 1$. 
	Again, by
	inductive hypothesis, let 
	$g_2$ be a walk such that $\alpha^{(k)}_{n{-}2} = \gamma^{g_2}_k$ and $g_2(1) = h_2$. Let $h'_2$ be the final value of $g_2$, and let $h_3$ be such that $|h_3 {-} h'_2| \leq 1$, and $\gamma_k[h_3] = 1$. 
	Proceed in the same way, obtaining walks $g_3, \dots, g_{k}$. Taking $f$ to be the result of concatenating $g_1, g_2, g_3, \dots, g_{k}$
	in the obvious fashion yields the desired walk. 
	
	If $2k >n \geq k{+}2$, then we have $\alpha^{(k)}_{n} = \alpha^{(k)}_{n{-}1} \alpha^{(k)}_{n{-}2} \cdots \alpha^{(k)}_{k{+}1}\beta_{k}\beta_{k{-}1} \cdots \beta_{k{-}p}$, where $p = 2k{-}n$.
	We begin as in the previous paragraph: setting $h_1 = h$, by inductive hypothesis, let $g_1$ be a walk such that $\alpha^{(k)}_{n{-}1} = \gamma^{g_1}_k$ and $g_1(1) = h_1$.
	Now let $h'_1$ be the final value of $g_1$, that is, $g_1(|\alpha^{(k)}_{n{-}1}|) = h'_1$. By Claim~\ref{claim:occursIn}, there exists
	$h_2$ such that $|h_2 {-} h'_1| \leq 1$, and $\gamma_k[h_2] = 1$. Now continue as before so as to obtain walks $g_2, g_3\dots$, with 
	respective starting points $h_2, h_3, \dots$,
	but stopping when we have obtained $g_{k{-}p{-}1}$, and the following starting point $h_{k{-}p}$. 
	Observe that concatenating $g_1, g_2, g_3, \dots, g_{k{-}p{-}1}$ gives a walk on $\gamma_k$ which yields the word  
	$\alpha^{(k)}_{n{-}1} \alpha^{(k)}_{n{-}2} \cdots \alpha^{(k)}_{k{+}1}$.
	By Claim~\ref{claim:kaybonacciStrollEarlyTerms},		
	choose $g_{k{-}p}$ to be a walk on $\gamma_k$ yielding the word
	$\beta_{k}\beta_{k-1} \cdots \beta_{k{-}p} = \beta_{k}\beta_{k-1} \cdots \beta_{k{-}p}$ and with $g_{k{-}p}(1) = h_{k{-}p}$.		
	Taking $f$ to be the result of concatenating $g_1, g_2, g_3, \dots, g_{k{-}p}$
	establishes the claim. 
\end{proof}
\subsubsection*{Acknowledgements}
This work was supported by the Polish NCN, grant number\linebreak
2018/31/B/ST6/03662.
The author wishes to thank Prof.~V.~Berth\'{e} and
\linebreak
Prof.~L.~Tendera for their valuable help, and Mr.~D.~Kojelis for his many suggestions, in particular the much-improved reformulation of Theorem~\ref{theo:defects}.
\bibliographystyle{plain}
\bibliography{wordLemmav2}
\end{document}